\documentclass[review,3pt,onecolumn]{elsarticle}
\usepackage[UKenglish]{babel} %
\usepackage[T1]{fontenc}
\usepackage[utf8]{inputenc}
\usepackage{amsmath}
\usepackage{amsthm}
\usepackage{amssymb}
\usepackage{mathtools}
\usepackage{thmtools}
\usepackage{commath}
\usepackage{bbm} %
\usepackage{verbatim} %
\usepackage[dvipsnames]{xcolor} %
\usepackage{algorithmic}
\usepackage{algorithm}
\usepackage{cleveref}
\usepackage{comment}
\usepackage[caption=false,font=footnotesize]{subfig}
\usepackage{array}
\usepackage{algorithmic}
\usepackage{algorithm}
\usepackage{mathrsfs}
\usepackage[]{siunitx}
\usepackage{microtype}

\allowdisplaybreaks

\DeclareUnicodeCharacter{00A0}{ } 
\DeclarePairedDelimiterX\Set[2]{\lbrace}{\rbrace}%
 { #1 \,:\, #2 } 

\newcommand{\diag}{\operatorname{diag}} 

\newcommand{\E}{\mathbb{E}}
\newcommand{\C}{\mathbb{C}}
\newcommand{\V}{\mathbb{V}}

\newcommand{\R}{\mathbb{R}} 

\makeatletter
	\@ifundefined{md}{%
		}
		{
			
		}
\makeatother

\newcommand{\T}{\ensuremath{\mathsf{T}}}

\usepackage{tikz}
\usepackage{pgfplots}
\usepgfplotslibrary{external}
\usepgfplotslibrary{statistics}
\usetikzlibrary{decorations}
\usetikzlibrary{decorations.markings}
\usetikzlibrary{pgfplots.statistics}
\pgfplotsset{compat = 1.3}

\definecolor{myred}{rgb}{0.82, 0.1, 0.26} %
\definecolor{mygreen}{rgb}{0.0, 0.7, 0.0} %
\definecolor{myblue}{rgb}{0,0.4,1} %
\definecolor{myyellow}{rgb}{1,0.8,0.4} %
\definecolor{mybrown}{RGB}{169,94,43} %
\definecolor{mycyan}{RGB}{29,202,190} %
\definecolor{myviolet}{RGB}{223,133,223} %
\definecolor{mypink}{RGB}{255,51,153} %
\definecolor{mygreen2}{RGB}{153,230,93} %

\definecolor{AaltoYellow}{RGB}{255, 205, 0}
\definecolor{AaltoRed}{RGB}{239, 51, 64}
\definecolor{AaltoBlue}{RGB}{0, 94, 184}
\definecolor{AaltoPurple}{RGB}{117, 59, 189}
\definecolor{AaltoGreen}{RGB}{120, 190, 32}
\definecolor{AaltoPurple2}{RGB}{192, 26, 162}

\newtheorem{proposition}{Proposition}

\newtheorem{remark}{Remark}

\newif\ifdouble %
 \doubletrue

\begin{document}

\begin{frontmatter}

\title{Iterative Statistical Linear Regression for Gaussian Smoothing in Continuous-Time Non-linear Stochastic Dynamic Systems}

\author{Filip~Tronarp and Simo~Särkkä}

\address{Department of Electrical Engineering and Automation, Aalto University, Rakentajanaukio 2, 02150 Espoo, Finland}

\begin{abstract}
This paper considers approximate smoothing for discretely observed non-linear stochastic differential equations. The problem is tackled by developing methods for linearising stochastic differential equations with respect to an arbitrary Gaussian process. Two methods are developed based on 1) taking the limit of statistical linear regression of the discretised process and 2) minimising an upper bound to a cost functional. Their difference is manifested in the diffusion of the approximate processes. This in turn gives novel derivations of pre-existing Gaussian smoothers when Method 1 is used and a new class of Gaussian smoothers when Method 2 is used. Furthermore, based on the aforementioned development the iterative Gaussian smoothers in discrete-time are generalised to the continuous-time setting by iteratively re-linearising the stochastic differential equation with respect to the current Gaussian process approximation to the smoothed process. The method is verified in two challenging tracking problems, a reentry problem and a radar tracked coordinated turn model with state dependent diffusion. The results show that the method has competitive estimation accuracy with state-of-the-art smoothers.  
\end{abstract}

\begin{keyword}
stochastic differential equations, statistical linear regression, iterative methods, continuous-discrete Gaussian smoothing.   
\end{keyword}

\end{frontmatter}

\section{Introduction}
Inference in continuous-time stochastic dynamic systems is a frequently occurring topic in disciplines such as navigation, tracking, and time series modelling \cite{TittertonWeston2004,CrassidisJunkins2004,LindstromMadsenNielsen2015,Sarkka2013}. The system is typically described in terms of a latent Markov process $\{X(s)\}_{s\geq 0}$, governed by a stochastic differential equation (SDE) \cite{Oksendal2003}. Furthermore, the process $\{X(s)\}_{s\geq 0}$ is assumed to be measured at a set of time instants $\{t_k\}_{k=1}^K$ by a collection of random variables $\{Y(t_k)\}_{k=1}^K$, each having a conditional distribution with respect to the process outcome at the corresponding time stamp. For the special case of an affine, Gaussian system, the calculation of predictive, filtering, and smoothing distributions amount to manipulating the joint moments of the latent process and the measurement process. In the case of filtering this procedure is known as Kalman-Bucy filtering \cite{KalmanBucy1961}. It was subsequently shown that the smoothing moments can be expressed in terms of ordinary differential equations with the filter moment as inputs \cite{RauchTungStriebel1965,Striebel1965}. 

While the theory of filtering and smoothing in linear dynamic systems is mature, the case of non-linear systems is still an area of intense research and a common strategy is to find a suitable linearisation of the system which enables the aforementioned methods for affine systems. An early approach was to linearise the system around the mean trajectory using truncated Taylor series \cite{Jazwinski:1970}. Another, fairly recent, approach was to apply the Rauch-Tung-Striebel smoother \cite{RauchTungStriebel1965} to a discretisation of the process and applying the limit $\delta t \to 0$ \cite{Sarkka2010}. The aforementioned approaches belong to the class of Gaussian smoothers, which were further studied in \cite{SarkkaSarmavuori2013}, where the non-linear smoothing theory \cite{Striebel1965,Leondes1970} was used to derive the Type I smoother. The Type I smoother of \cite{SarkkaSarmavuori2013} was subsequently re-derived in \cite{Koyama2018} using the projection methods developed in \cite{Brigo1998,Brigo1999}.

Another line of research has been into variational Gaussian smoothers, which imposes a Gaussian process posterior by fixed-form variational Bayes \cite{Archambeau2007}. This class of smoothers operate by iteratively refining the approximate smoothing solutions by taking gradient steps associated with the evidence lower bound. However, the development in \cite{Archambeau2007} requires a non-singular, state-independent diffusion. The variational smoother of \cite{Archambeau2007} was extended in \cite{Ala-Luhtala2015} by allowing singular diffusions if the drift in the singular subspace is an affine function of the state. Nonetheless, both formulations of variational smoothers require a state independent diffusion \cite{Archambeau2007,Ala-Luhtala2015}. The requirement of state independent diffusion was later circumvented by relaxing the condition of the fixed-form posterior to only have Gaussian marginals in \cite{Sutter2016}. However, this adds the requirement of solving a stochastic partial differential equation to obtain an initial value to their optimal control formulation \cite{Sutter2016}. This can be computationally challenging even for moderately large state spaces, which is why the method was only validated on problems of small state space (one dimension). Another approach in a similar vein is the continuous-time expectation propagation algorithm \cite{Cseke2016}, which is also applicable when the diffusion is state dependent. However, accurate approximations of the moments with respect to the so-called \emph{tilted distributions}, which becomes a non-trivial problem as the dimension of the state-space grows. 

A recent advance, in discrete time inference, is the iterated posterior linearisation smoother (IPLS) \cite{Garcia2017} (see also \cite{Tronarp2018a}) which generalises the iterated extended Kalman smoother \cite{Bell1994} to sigma-point methods. This was done on the basis of statistical linear regression \cite{Lefebvre:2002}, where a given smoothing solution is improved upon by re-linearising the system using the current Gaussian smoother approximation and then running the smoother again \cite{Garcia2017,Tronarp2018a}. However, an analogue for continuous-time smoothing has yet to appear. 

The purpose of this paper is thus to generalise the discrete-time smoother of \cite{Garcia2017,Tronarp2018a} to the continuous-time case. This is accomplished by generalising the statistical linear regression method \cite{Lefebvre:2002} to the setting of stochastic differential equations. Two ways of doing this are discovered, 1)  
taking the limit of the statistical linear regression solution of a discretised process and 2) minimising an upper bound to a cost functional. This gives novel derivations of the smoothers in \cite{SarkkaSarmavuori2013} when Method 1 is used and a new kind of smoothers when Method 2 is used. Furthermore, using the aforementioned linearisation methods iterative Gaussian smoothers, analogous to the discrete-time iterative Gaussian smoother \cite{Bell1994,Garcia2017,Tronarp2018a}, are developed for stochastic differential equations.

In \Cref{sec:problem_formulation} the smoothing problem is formally posed, and linear smoothing theory and previous approaches to smoothing in non-linear systems is reviewed. Lastly, the present contribution is outlined. In \Cref{sec:slr_sde}, the statistical linear regression method is generalised to stochastic differential equations. It is derived both as a discrete time limit and as a minimiser to a certain cost functional. The development in \Cref{sec:slr_sde} is subsequently combined with a linear smoothing theory to arrive at novel derivations of the smoothers presented in \cite{Sarkka2010,SarkkaSarmavuori2013}. The main result is presented in \Cref{sec:iterative_smoothers} where the discrete time iterative Gaussian smoothers \cite{Garcia2017,Tronarp2018a,Bell1994} are generalised to continuous-time models. In \Cref{sec:experimental_results} the iterative smoothers are demonstrated in two non-linear and high-dimensional target tracking problems. The manuscript ends with the conclusion in \Cref{sec:conclusion}.

\section{Problem Formulation}\label{sec:problem_formulation}
The setting is as follows, there is a latent Markov process $\{X(t)\}_{t\geq 0}, \ X(t) \in \mathbb{R}^{d_X}$ which is assumed to evolve according to the following discretely observed stochastic differential equation (SDE) model
\begin{subequations}\label{eq:general_model}
\begin{align}
&\dif X(t) = \mu(t,X(t))\dif t + \sigma(t,X(t))\mathrm{d}W(t), \label{eq:dynamics} \\
&Y(t_k) = h(t_k,X(t_k)) + V(t_k),  \label{eq:measurements} \\
&\C[V(t_k),V(t_l)] =\delta_{k,l}R \nonumber , \ V(t_k) \sim \mathcal{N}(0,R), \nonumber
\end{align}
\end{subequations}
where $f_{X(0)}(x) = \mathcal{N}(x;\bar{x}(0^-),\Sigma(0^-) )$, $\mu \colon \R_+\times\R^{d_X} \to \R^{d_X}$ is a drift function, $\sigma \colon \R_+\times\R^{d_X} \to \R^{d_X\times d_W}$ is a diffusion matrix, $\{W(s)\}_{s\geq 0}$ is a $d_W$-dimensional standard Brownian motion, $h \colon \R_+\times\R^{d_X} \to \R^{d_Y}$ is a measurement function, and $\delta_{k,l}$ is Kronecker's delta function.  Furthermore, given a measurement series $\{y(t_k)\}_{k=1}^K, \ t_{k+1} > t_k$. The set of measurements up to just before time $t$ and the set of measurements up to precisely time $t$ are denoted by $\mathscr{Y}({t^-}) = \{y(t_k) \colon  t_k < t \}$ and $\mathscr{Y}(t) =  \{y(t_k) \colon  t_k \leq t \}$, respectively.

The inference problem for $X(\tau)$ is then in the Bayesian sense to find a family of conditional densities 
\begin{equation}\label{eq:generic_pomp_inference}
f_{X(\tau) \mid \mathscr{Y}({t_k})}(x) , \ k= 1,\dots,K .
\end{equation}
When $\tau < t_k$ the probability density function in \Cref{eq:generic_pomp_inference} is said to be a smoothing distribution, if, in particular, $\tau = t_k$ it is a filtering distribution and if $\tau > t_k$ then the density in \Cref{eq:generic_pomp_inference} is said to be a predictive distribution. Moreover, the expectation, cross-covariance, and covariance operators are denoted by $\E[\cdot], \ \C[\cdot,\cdot]$ and $\V[\cdot]$. We use the following notation:
\begin{subequations}
\begin{align}
\bar{x}(t) &= \E[X(t) \mid \mathscr{Y}(t)], \\
\bar{x}(t^-) &= \E[X(t) \mid \mathscr{Y}(t^-)], \\
\Sigma(t) &= \V[X(t) \mid \mathscr{Y}(t)], \\
\Sigma(t^-) &= \V[X(t) \mid \mathscr{Y}(t^-)],
\end{align}
\end{subequations}
and similarly for the smoothing moments based on the entire measurement series:
\begin{subequations}
\begin{align}
\hat{x}(t) &= \E[X(t) \mid \mathscr{Y}(t_K)], \\
\Omega(t) &= \V[X(t) \mid \mathscr{Y}(t_K)].
\end{align}
\end{subequations}

\subsection{Prior Work}
Smoothing in state space models has endured long and considerable efforts in the past 50 years \cite{RauchTungStriebel1965,Striebel1965,Leondes1970,Kallianpur1968}. First the linear smoothing theory will be reviewed and subsequently the more prominent approaches to approximate smoothers. 

\subsubsection{Linear smoothing theory}
The linear smoothing theory applies to systems of the following form:
\begin{subequations}\label{eq:affine_model}
\begin{align}
&\dif X(t) = ( A(t) X(t) + b(t) ) \dif t + \sigma(t) \dif W(t), \label{eq:affine_dynamics}  \\
&Y({t_k}) = C(t_k) X(t_k) + d({t_k}) + V(t_k), \\ 
&\C[V(t_k),V(t_l)] =\delta_{t_k,t_l}R, \ V(t_k) \sim \mathcal{N}(0,R), \nonumber
\end{align}
\end{subequations}
where $W(t)$ is a standard Wiener process. Since the collection $(X(t_{1:K}),Y(t_{1:K}))$ is jointly Gaussian the conditioning reduces to projections in a finite dimensional space. Furthermore, this can be implemented in a sequential manner where alternations between update and predictions are carried out \cite{KalmanBucy1961}.

Starting with the predictive distribution, $X(t_k^-) \sim \mathcal{N}(\bar{x}(t_k^-),\Sigma(t_k^-) )$, the parameters of the filtering distribution are then computed according to \cite{Sarkka2013}
\begin{subequations}\label{eq:kalman_update}
\begin{align}
S(t_k) &=  C(t_k) \Sigma(t_k^-) C(t_k)^\T +  R, \\
K(t_k) &= \Sigma(t_k^-)C(t_k)^\T S^{-1}(t_k),  \\
\hat{v}(t_k) &= y(t_k) - C(t_k) \bar{x}(t_k^-) - d(t_k), \\
\bar{x}(t_k) &=  \bar{x}(t_k^-) + K(t_k) \hat{v}(t_k), \\
\Sigma(t_k) &= \Sigma(t_k^-) - K(t_k) S(t_k) K(t_k)^\T .
\end{align}
\end{subequations}
The predictive distribution at $t_{k+1}^-$ is then given by solving \cite{Oksendal2003,KalmanBucy1961}
\begin{subequations}\label{eq:kalman_prediction}
\begin{align}
\frac{\dif \bar{x}(t)}{\dif t} &= A(t)\bar{x}(t) + b(t),  \\
\frac{\dif \Sigma(t)}{\dif t}  &= A(t)\Sigma(t) + \Sigma(t) A(t)^\T + Q(t),
\end{align}
\end{subequations}
on the interval $\tau \in [t_k,t_{k+1}]$ with initial conditions $(\bar{x}(t_k),\Sigma(t_k) )$, where $Q(t) = \sigma(t)\sigma(t)^\T$. The differential equations for the smoothing moments can then be expressed in terms of the filtering moments according to \cite{RauchTungStriebel1965,Striebel1965,Leondes1970} 
\begin{subequations}\label{eq:linear_smoothing_equation}
\begin{align}
\begin{split}
\frac{\dif \hat{x}(t)}{\dif t} =& A(t)\hat{x}(t) + b(t) \\
&+Q(t)\Sigma^{-1}(t)( \hat{x}(t) - \bar{x}(t)),
\end{split}\\
\begin{split}
\frac{\dif \Omega(t)}{\dif t}   =& [A(t) + Q(t)\Sigma^{-1}(t)]\Omega(t) \\
&+ \Omega(t)[A(t) + Q(t)\Sigma^{-1}(t)]^\T - Q(t). 
\end{split}
\end{align}
\end{subequations} 
\subsubsection{The Non-linear smoothing theory approach}
Now consider the smoothing problem for the non-linear model in \Cref{eq:general_model}. The expectations with respect to the smoothing distribution of some test function $\psi(X(t))$ was studied in \cite{Striebel1965,Leondes1970}, and the backwards differential equation 
\begin{equation*}\label{eq:smoothing_moment_equation}
 \frac{\dif }{\dif t} \E[\psi(X(t))\mid \mathscr{Y}(t_K)] = \E[\mathscr{K}_1 \psi(X(t))\mid \mathscr{Y}(t_K)]
\end{equation*}  
is obtained, where the operator $\mathscr{K}_1$ is given by 
\begin{equation*}
\begin{split}
&\mathscr{K}_1\psi(x(t) ) = \sum_p \mu_p(t,x(t)) \partial_{p}\psi(x(t)) \\
&- \sum_{p,r}\partial_{p} \psi(x(t)) \partial_{r}[Q(t,x(t))]_{p,r}   \\
&- \frac{1}{2}\sum_{p,r} [Q(t,x(t))]_{p,r}\partial^2_{p,r} \psi(x(t)) \\
&- \frac{\sum_{p,r} [Q(t,x(t))]_{p,r} \partial_{p}\psi(x(t))\partial_{r}\psi(x(t))}{f_{X(t)\mid \mathscr{Y}(t)}(x(t))} ,
\end{split}
\end{equation*}
where $\partial_{r}$ is the partial derivative operator with respect to the $r$:th coordinate in $X$ and $\partial_{p,r}^2$ is the composition of  $\partial_{r}$ and  $\partial_{p}$. Approaches to implement this was not covered in \cite{Striebel1965} and in \cite{Leondes1970} a Taylor expansion was used. However, by plugging in a Gaussian approximation to $f_{X(t)\mid \mathscr{Y}(t)}(x(t))$, the Type I smoother of \cite{SarkkaSarmavuori2013} is derived. The Type II and III smoothers are also discussed by \cite{SarkkaSarmavuori2013}, which originate from the smoother formulation developed in \cite{Sarkka2010}.  

\subsubsection{The Kullback-Leibler approach}
Another approach to smoothing in non-linear systems is based on minimising the Kullback-Leibler divergence between the posterior measure, $\mathbb{P}_{X\mid Y}$, and a fixed form Gaussian measure, $\mathbb{Q}_{X\mid Y}$ \cite{Archambeau2007}. This requires (i) $\sigma(t,X(t)) = \sigma(t)$ and (ii)  $Q(t)^{-1} = (\sigma(t)\sigma^\T(t))^{-1}$ exists. The Kullback-Leibler divergence is given by   
\begin{equation}\label{eq:kld_vb}
\begin{split}
\operatorname{KL}[\mathbb{Q}_{X\mid Y} \mid \mid \mathbb{P}_{X\mid Y} ] =& \frac{1}{2}\int_0^T E(t) \dif t + \frac{K d_X}{2}\\
&+ \frac{K}{2}\log \det R + \log Z,
\end{split}
\end{equation}
where $Z$ is a normalisation constant, $E(t) = E_X(t) + E_{Y\mid X}(t)$, and 
\begin{subequations}
\begin{align*}
\begin{split}
E&_X(t) \\
&= \E^\mathbb{Q}\Big[ \lvert \lvert \mu(t,X(t)) - A(t)X(t) - b(t) \rvert\rvert_{Q^{-1}(t)}^2\Big], 
\end{split}\\
\begin{split}
E&_{Y\mid X}(t)  \\
&=\sum_{k=1}^K \E^\mathbb{Q}\Big[ \lvert \lvert y(t_k) - h(t_k,X(t_k)) \rvert\rvert_{R^{-1}}^2\Big] \delta(t-t_k),
\end{split}
\end{align*}
\end{subequations}
where $\lvert \lvert \cdot \rvert\rvert_W$ is a weighted Euclidean norm with weighting matrix $W$. Now, it is clear that the above functional is not well defined if $Q(t)$ is singular. This was extended in \cite{Ala-Luhtala2015} to the case of singular $Q(t)$ under the assumption that the drift function of the singular sub-space is an affine function of the state. This was further extended to the case of state dependent $Q$ in \cite{Sutter2016}, where they minimise \Cref{eq:kld_vb} subject to $\mathbb{Q}_{X\mid Y}$ being a diffusion process with a prescribed marginal law (i.e Gaussian). However, this requires computing $w(x,0)$, with $w(x,t)$ being the solution to the Kolmogorov backward equation.  

Furthermore, the minimisation objective of continuous-time expectation propagation\cite{Cseke2016} is retrieved by exponentiating the argument to the expectations in the definition of $E_{Y\mid X}(t)$, along with some other minor modifications.

\subsection{The Contribution}
The aim of this paper is to develop iterative techniques for obtaining smoothing estimates of the system in \Cref{eq:dynamics,eq:measurements}. More specifically the following contributions are put forth:
\begin{enumerate}
\item The statistical linear regression~\cite{Lefebvre:2002,Garcia2015} method is generalised to the setting of stochastic differential equations. Two alternative versions of this procedure are provided, 1) is based on applying standard statistical linear regression to a discretised version of an auxiliary SDE and passing to the continuous limit, while 2) sets up a least squares problem for the difference between the auxiliary SDE and an affine approximation, which results in the minimisation of a quadratic functional, which is solved by methods in variational calculus~\cite{Weinstock1974}. \label{item:1}

\item It is shown that the Type II and III smoothers of \cite{SarkkaSarmavuori2013} can be derived by using Method 1 for linearising the stochastic differential equation together with the linear smoothing results \cite{RauchTungStriebel1965,Striebel1965}. Using Method 2 gives a new class of smoothers. These shall be separated by the suffixes \emph{of the first kind} and \emph{of the second kind} for Method 1 and Method 2 of linearising the stochastic differential equation, respectively. 

\item Iterated smoothers are developed on the basis of using Contribution \ref{item:1} above to re-linearise the SDE with respect to the current best Gaussian process approximation of the smoothed process, which gives a continuous-time generalisation of the iterative Gaussian smoothers \cite{Garcia2017,Tronarp2018a}. 
\end{enumerate}

\section{Statistical Linear Regression For Stochastic Differential Equations}\label{sec:slr_sde}
In this section, the statistical linear regression method \cite{Lefebvre:2002} (see also \cite{Garcia2015}) is generalised to the case of affine approximations of stochastic differential equations. Let $X(t)$ be driven by the SDE in \Cref{eq:dynamics}. Then a Gaussian process, $\{\widehat{X}(s)\}_{s\geq 0}$ can be used to approximate the evolution of $X(t)$ according to 
\begin{equation}\label{eq:sde_approximation}
\dif X(t) \approx [A(t)X(t) + b(t)] \dif t + \bar{\sigma}(t) \dif \widehat{W}(t),  
\end{equation}
where $\widehat{W}(t)$ is a standard Wiener process. The procedures for doing this is given in \Cref{alg:slrsde,alg:slrsde2}. The role of the Gaussian process, $\{\widehat{X}(s)\}_{s\geq 0}$, in \Cref{alg:slrsde,alg:slrsde2} is essentially, to approximate the drift function and the diffusion matrix according to \Cref{eq:sde_approximation} This is done by considering an auxiliary process defined by 
\begin{equation}\label{eq:auxiliary_process}
\dif \widetilde{X}(t) = \mu(t,\widehat{X}(t)) \dif t + \sigma(t,\widehat{X}(t)) \dif \widetilde{W}(t),
\end{equation}
where $\widetilde{W}(t)$ is a standard Brownian motion with that is independent of $\widehat{X}(t)$. The parameters $A(t)$, $b(t)$, $\bar{\sigma}(t)$ are then found by making an affine approximation of \Cref{eq:auxiliary_process} in terms of $\widehat{X}(t)$. This can be done in two ways, (i) performing SLR \cite{Lefebvre:2002} on an Euler-Maruyama discretisation of \Cref{eq:auxiliary_process} and passing to the limit, which gives \Cref{alg:slrsde}, and (ii) using \Cref{eq:auxiliary_process} to set up a cost functional for $A(t)$, $b(t)$, and $\bar{\sigma}(t)$, which results in \Cref{alg:slrsde2}. Detailed derivations of these approaches are developed in \Cref{subsec:slr_discrete_time_limit,subsec:slr_variational_calculus}.  

Furthermore, it should be noted that \Cref{alg:slrsde,alg:slrsde2} suggests an iterative scheme for smoothing, analogous to the discrete-time case \cite{Garcia2017,Tronarp2018a}. That is, the linearising process $\widehat{X}(t)$ is taken to be the current best approximation of the smoothing solution of \Cref{eq:general_model}. The parameters that are retrieved can then be used in conjunction with \Cref{eq:kalman_update,eq:kalman_prediction,eq:linear_smoothing_equation} to retrieve a new approximate smoothing solution. This issue that shall be revisited in \Cref{sec:iterative_smoothers}. 

\begin{algorithm}
\caption{Statistical Linear Regression I \newline(Discrete-time limit)}
\label{alg:slrsde}
\begin{algorithmic}
\REQUIRE The  marginal moment functions of a Gaussian process $\widehat{X}(s)$, $\{(\E[\widehat{X}(s)], \V[\widehat{X}(s)]) \}_{s=0}^T$, drift function $\mu(t,X(t))$, and diffusion matrix $\sigma(t,X(t))$.  
\ENSURE Approximate drift function, $A(t)X(t) + b(t)$, and a diffusion matrix $\bar{\sigma}_1(t)$. 
\STATE $A(t) \leftarrow \C[\mu(t,\widehat{X}(t)),\widehat{X}(t)]\V[\widehat{X}(t)]^{-1}$
\STATE $b(t) \leftarrow \E[\mu(t,\widehat{X}(t))] - A(t)\E[\widehat{X}(t)]$
\STATE $\bar{\sigma}_1(t) \leftarrow \E[\sigma(t,\widehat{X}(t))\sigma^\T(t,\widehat{X}(t))]^{1/2}$
\end{algorithmic}
\end{algorithm}

\begin{algorithm}
\caption{Statistical Linear Regression II \newline(functional minimisation)}
\label{alg:slrsde2}
\begin{algorithmic}
\REQUIRE The  marginal moment functions of a Gaussian process $\widehat{X}(s)$, $\{(\E[\widehat{X}(s)], \V[\widehat{X}(s)]) \}_{s=0}^T$, drift function $\mu(t,X(t))$, and diffusion matrix $\sigma(t,X(t))$.  
\ENSURE Approximate drift function, $A(t)X(t) + b(t)$, and a diffusion matrix $\bar{\sigma}_2(t)$. 
\STATE $A(t) \leftarrow \C[\mu(t,\widehat{X}(t)),\widehat{X}(t)]\V[\widehat{X}(t)]^{-1}$
\STATE $b(t) \leftarrow \E[\mu(t,\widehat{X}(t))] - A(t)\E[\widehat{X}(t)]$
\STATE $\bar{\sigma}_2(t) \leftarrow \E[\sigma(t,\widehat{X}(t))]$
\end{algorithmic}
\end{algorithm}
\begin{remark}\label{rem:short_term_slr}
If $X(\tau) \sim \mathcal{N}(\E[X(\tau)],\V[X(\tau)])$ then \Cref{alg:slrsde,alg:slrsde2} can be used to obtain a Gaussian approximation of $X(\tau+\delta)$ by approximating the drift function and diffusion matrix at time $\tau$ using the moments of $X(\tau)$. This is the usual procedure in continuous-time Gaussian filtering (cf. \cite{SarkkaSarmavuori2013}), where \Cref{alg:slrsde} has previously been used implicitly. 
\end{remark}
\subsection{Discrete Time Limit}\label{subsec:slr_discrete_time_limit}
Here, a short-term variant of the procedure in \Cref{alg:slrsde} is derived (see \Cref{rem:short_term_slr}). Using an Euler-Maruyama discretisation \cite{KloedenPlaten1999} of \Cref{eq:auxiliary_process} gives 
\begin{equation}
\widetilde{X}(t+\delta) = \widetilde{X}(t) +  \mu(t,\widehat{X}(t)) \delta + \sigma(t,\widehat{X}(t)) \delta \widetilde{W}(t),
\end{equation} 
where $\delta \widetilde{W}(t) = \widetilde{W}(t+\delta) - \widetilde{W}(t)$ is a standard Wiener increment of size $\delta$. The statistical linear regression approach then involves forming an affine approximation to $\widetilde{X}(t+\delta)$ according to 
\begin{equation}
\widetilde{X}_a(t+\delta) = \widetilde{X}_a(t) + [A(t)\widehat{X}(t) + b(t) ] \delta + \Xi(t,\delta),
\end{equation}
where $ \Xi(t,\delta)$ is a zero mean random variable with covariance matrix $\Gamma(t,\delta)$ accounting for the error, assumed to be Gaussian (see e.g \cite{Garcia2015}), and $\widetilde{X}_a(t) \triangleq \widetilde{X}(t)$. The parameters $A(t)$ and $b(t)$ are then found by minimising the mean squared error of the residual: 
\begin{subequations}
\begin{align*}
\E\Bigg[\norm{\widetilde{X}(t+\delta) - \widetilde{X}_a(t) - [A(t)\widehat{X}(t) + b(t) ] \delta }^2\Bigg].  
\end{align*}
\end{subequations} 
This is simply a quadratic optimisation problem and the parameters $A(t)$ and $b(t)$ are thus (c.f \cite{Lefebvre:2002,Garcia2015})
\begin{subequations}
\begin{align*}
A(t) &=\C[\mu(t,\widehat{X}(t)),\widehat{X}(t)]\V[\widehat{X}(t)]^{-1}, \\
b(t) &= \E[\mu(t,\widehat{X}(t))] - A(t)\E[\widehat{X}(t)]. \\
\end{align*}
\end{subequations}
Furthermore, the residual is given by
\begin{equation*}
\widetilde{R}(t,\delta) = \widetilde{X}(t+\delta) - \widetilde{X}_a(t) - [A(t)\widehat{X}(t) + b(t) ] \delta.
\end{equation*}
Straight-forward calculations gives the moments of $\widetilde{R}(t,\delta)$ as 
\begin{subequations}
\begin{align*}
\E[\widetilde{R}(t,\delta)] &= 0, \\
\V[\widetilde{R}(t,\delta)] &= \E[\sigma(t,\widehat{X}(t))\sigma^\T(t,\widehat{X}(t))]\delta + o(\delta),
\end{align*}
\end{subequations}
which are taken to be the moments of $\Xi(t,\delta)$. That is, 
\begin{equation}
\Gamma(t,\delta) = \E[\sigma(t,\widehat{X}(t))\sigma^\T(t,\widehat{X}(t))]\delta + o(\delta). 
\end{equation}
Now define $\bar{\sigma}_1(t) = \E[\sigma(t,\widehat{X}(t))\sigma^\T(t,\widehat{X}(t))]^{1/2}$. Then the increment of $\widetilde{X}_a(t)$ is approximately given by 
\begin{equation*}
\begin{split}
\widetilde{X}_a(t+\delta) - \widetilde{X}_a(t) &= \Big(A(t) \widetilde{X}_a(t) + b(t) \Big)\delta \\
&\quad+ \bar{\sigma}_1(t) \delta\widehat{W}(t)+ o(\delta), 
\end{split}
\end{equation*}
where $\delta\widehat{W}(t) = \widehat{W}(t+\delta) - \widehat{W}(t)$ is a standard Wiener increment of size $\delta$, independent of $\widetilde{X}_a(t)$ and $W(t)$, that matches the variance of the part of $\Xi(t,\delta)$ that does not vanish faster than $\delta$ as $\delta \to 0$. Now, passing to the limit, $\delta \to 0$, gives the following differential
\begin{equation*}
\dif \widetilde{X}_a(t) = \Big(A(t) \widetilde{X}_a(t) + b(t) \Big)\dif t + \bar{\sigma}_1(t) \dif\widehat{W}(t),
\end{equation*}
which makes the procedure in \Cref{alg:slrsde} apparent.

\subsection{A Variational Formulation}\label{subsec:slr_variational_calculus}
Another approach for arriving at \Cref{eq:sde_approximation} is by employing variational calculus \cite{Weinstock1974} as follows. Define an approximating process to \Cref{eq:auxiliary_process}, $\widetilde{X}_a(t)$, with $\widetilde{X}_a(t_b) = \widetilde{X}(t_b)$ and $t_b > 0$, given by 
\begin{equation}
\dif \widetilde{X}_a(t) = (A(t)\widehat{X}(t) + b(t)) \dif t + \bar{\sigma}_2(t) \dif \widetilde{W}(t). 
\end{equation}
For $t_e>t_b$, the mean square error between $\widetilde{X}_a(t_e)$ and $\widetilde{X}(t_e)$ is given by
\begin{equation*}
\begin{split}
\E[\lvert\lvert &\widetilde{X}_a(t_e) -  \widetilde{X}(t_e) \rvert\rvert^2] \\
=&  \E\Bigg[\Big\lvert\Big\lvert \int_{t_b}^{t_e} A(\tau)\widehat{X}(\tau) + b(\tau) - \mu(\tau,\widehat{X}(\tau)) \dif \tau  \Big\rvert\Big\rvert^2\Bigg] \\
&+ \E\Bigg[\Big\lvert\Big\lvert \int_{t_b}^{t_e} \bar{\sigma}_2(\tau) -  \sigma(\tau,\widehat{X}(\tau)) \dif \widetilde{W}(\tau)  \Big\rvert\Big\rvert^2\Bigg] 
\end{split}
\end{equation*}
where the independence between $\widehat{X}(t)$ and $\widetilde{W}(t)$ was used to eliminate the cross-term. Furthermore, employing Jensen's inequality and It\^o isometry gives
\begin{equation*}
\begin{split}
&\E[\lvert\lvert \widetilde{X}_a(t_e) -  \widetilde{X}(t_e) \rvert\rvert^2] \\
&\quad \leq \E\Bigg[ \int_{t_b}^{t_e} \big\lvert\big\lvert A(\tau)\widehat{X}(\tau) + b(\tau) - \mu(\tau,\widehat{X}(\tau))\big\rvert\big\rvert^2 \dif \tau  \Bigg] \\
&\qquad + \E\Bigg[ \int_{t_b}^{t_e} \lvert\lvert \bar{\sigma}_2(\tau) - \sigma(\tau,\widehat{X}(\tau)) \rvert\rvert_F^2 \dif \tau  \Bigg], 
\end{split}
\end{equation*}
where $\lvert\lvert \cdot \rvert\rvert_F$ is the Frobenius norm. Therefore, an appropriate cost functional for fitting $A$, $b$, and $\bar{\sigma}_2$ may be defined as 
\begin{equation}\label{eq:cost_functional}
\begin{split}
&\mathscr{J}(A,b,\bar{\sigma}_2) = \frac{1}{2}  \int_{t_b}^{t_e} \E\Big[\lvert\lvert \bar{\sigma}_2(\tau) - \sigma(\tau,\widehat{X}(\tau)) \rvert\rvert_F^2\Big] \dif \tau   \\
&\quad + \frac{1}{2}  \int_{t_b}^{t_e} \E\Big[\lvert\lvert A(\tau)\widehat{X}(\tau) + b(\tau) - \mu(\tau,\widehat{X}(\tau))\rvert\rvert^2\Big] \dif \tau   
\end{split}
\end{equation}
Perturbing $A$, $b$, and $\bar{\sigma}_2$ by arbitrary functions $\varepsilon_b$, $\varepsilon_A$, and $\varepsilon_{\bar{\sigma}_2}$, respectively gives
\begin{subequations}
\begin{align*}
\begin{split}
&\mathscr{J}(A+\varepsilon_A,b,\bar{\sigma}_2) - \mathscr{J}(A,b,\bar{\sigma}_2) = r_A(\varepsilon_A) \\
&+ \int_{t_b}^{t_e} \operatorname{tr}\{ \E[\widehat{X}(\tau)(\widehat{X}^\T(\tau)A^\T(\tau) + b^\T(\tau)) )\varepsilon_A(\tau) ]   \} \dif \tau\\
&- \int_{t_b}^{t_e} \operatorname{tr}\{ \E[\widehat{X}(\tau) \mu^\T(\tau,\widehat{X}(\tau)) \varepsilon_A(\tau) ]   \} \dif \tau
\end{split}\\
\begin{split}
&\mathscr{J}(A,b+\varepsilon_b,\bar{\sigma}_2) - \mathscr{J}(A,b,\bar{\sigma}_2) =  r_b(\varepsilon_b) \\
&\quad+ \int_{t_b}^{t_e} \E[A(\tau) \widehat{X}(\tau) + b(\tau) - \mu(\tau,\widehat{X}(\tau))]^\T \varepsilon_b(\tau) \dif \tau
\end{split}\\
\begin{split}
&\mathscr{J}(A,b,\bar{\sigma}_2+\varepsilon_{\bar{\sigma}_2}) - \mathscr{J}(A,b,\bar{\sigma}_2) =r_{\bar{\sigma}_2}(\varepsilon_{\bar{\sigma}_2})    \\
&\quad+ \int_{t_b}^{t_e} \operatorname{tr}\{ (\bar{\sigma}_2(\tau) - \E[\sigma(\tau,\widehat{X}(\tau))] \} \dif \tau,
\end{split}
\end{align*}
\end{subequations}
where $r_A$, $r_b$, and $r_{\bar{\sigma}_2}$ contain the higher order terms of $\varepsilon_A$, $\varepsilon_b$, and $\varepsilon_{\bar{\sigma}_2}$, respectively. The sufficient conditions for minima are given by \cite{Weinstock1974} :  
\begin{subequations}
\begin{align*}
0 &= \E[A(\tau)\widehat{X}(\tau) + b(\tau) - \mu(\tau,\widehat{X}(\tau))], \\
0 &= \E[\widehat{X}(\tau)(\widehat{X}^\T(\tau)A^\T(\tau) + b^\T(\tau)  -\mu^\T(\tau,\widehat{X}(\tau))  ) ], \\
0 &= \bar{\sigma}_2(\tau) - \E[\sigma(\tau,\widehat{X}(\tau))].
\end{align*}
\end{subequations}
Therefore, the minimisers are given by 
\begin{subequations}
\begin{align}
A(\tau) &= \C[\mu(\tau,\widehat{X}(\tau)),\widehat{X}(\tau)]\V[\widehat{X}(\tau)]^{-1}, \\ 
b(\tau) &= \E[\mu(\tau,\widehat{X}(\tau))] - A(\tau) \E[\widehat{X}(\tau)], \\
\bar{\sigma}_2(\tau) &= \E[\sigma(\tau,\widehat{X}(\tau))].
\end{align}
\end{subequations} 
Thus the procedure in \Cref{alg:slrsde2} is obtained. 

\begin{remark}
Note that the cost functional in \Cref{eq:cost_functional} is defined on any interval $[t_b,t_e]$ with $t_e > t_b$. Thus one can take $t_b = 0$ and $t_e = t_K$ to make \Cref{alg:slrsde2} globally defined. This is in contrast to \Cref{alg:slrsde} which is defined by stitching together local approximations. 
\end{remark}

\subsection{The difference between the discrete-time limit and the variational formulation}\label{subsec:dtl_vs_vf}
While the difference between the discrete-time limit approach (\Cref{alg:slrsde}) and the variational approach (\Cref{alg:slrsde2}) is small it is nonetheless interesting to highlight. The only difference being the diffusion matrices, $\bar{\sigma}_1(t)$ and $\bar{\sigma}_2(t)$ for \Cref{alg:slrsde} and \Cref{alg:slrsde2}, respectively. The following holds. 
\begin{proposition}\label{prop:K1vK2}
Assume the same Gaussian process, $\widehat{X}(t)$, is used to obtain $\bar{\sigma}_1(t)$ and $\bar{\sigma}_2(t)$ then the following inequality holds 
\begin{equation}
\operatorname{tr}\{ \bar{\sigma}_1(t)\bar{\sigma}_1^\T(t) \} \geq \operatorname{tr}\{ \bar{\sigma}_2(t)\bar{\sigma}_2^\T(t) \}.
\end{equation}
\end{proposition}
\begin{proof}
Plugging in the expressions from \Cref{alg:slrsde,alg:slrsde2} gives
\begin{equation}
\begin{split}
&\operatorname{tr}\{ \E[\sigma(t,\widehat{X}(t))\sigma^\T(t,\widehat{X}(t))] \} \\
&\quad\geq  \operatorname{tr}\{ \E[\sigma(t,\widehat{X}(t))]\E[\sigma(t,\widehat{X}(t))]^\T \},
\end{split}
\end{equation}
which is Jensen's inequality. 
\end{proof}

The result in \Cref{prop:K1vK2} essentially says that $\bar{\sigma}_2$ will be smaller than $\bar{\sigma}_1$, in Frobenius sense. Equality can be retrieved when $\sigma(t,X(t)) = \sigma(t)$ or when $\V[\widehat{X}(t)] \to 0$. Furthermore, it is clear that approximate implementations of \Cref{alg:slrsde,alg:slrsde2}, employing Taylor series expansions of the integrand up to first order around $\E[\widehat{X}(t)]$ also give $\bar{\sigma}_1 = \bar{\sigma}_2$.

\section{Continuous-Discrete Gaussian Smoothers}\label{sec:smoothers}
The linearisation technique presented in \Cref{sec:slr_sde} allows for the formulation of approximate smoothers to the system \Cref{eq:general_model} by simply plugging in $A(t)$, $b(t)$, and $\bar{Q}_i(t) = \bar{\sigma}_i(t)\bar{\sigma}_i^\T(t),\ i \in \{1,2\}$ into the linear smoothing equations in \Cref{eq:linear_smoothing_equation}. Furthermore, it has not been specified what Gaussian process, $\widehat{X}(t)$ is used to compute $A(t),\ b(t)$ and $\bar{\sigma}_i(t)$. 

In any case, the smoothers developed in this paper rely on a Gaussian approximation to the filtering distribution which can be constructed on the fly during filtering by using \Cref{alg:slrsde} or \Cref{alg:slrsde2} to compute the linearisation parameters (see \Cref{rem:short_term_slr}), which are then given by 
\begin{subequations}\label{eq:filter_linearisation}
\begin{align}
A(t) &= \C[\mu(t,X(t)),X(t)\mid \mathscr{Y}(t)]\Sigma^{-1}(t), \\
b(t) &= \E[\mu(t,X(t))\mid \mathscr{Y}(t)] - A(t)\bar{x}(t), \\
\bar{\sigma}_i(t) &=
\begin{cases}
\E[\sigma(t,X(t))\sigma^\T(t,X(t))\mid \mathscr{Y}(t)]^{1/2}, \ i = 1, \\
\E[\sigma(t,X(t))\mid \mathscr{Y}(t)], \ i = 2.
\end{cases}
\end{align}
\end{subequations}
Plugging \Cref{eq:filter_linearisation} into \Cref{eq:kalman_prediction} gives
\begin{subequations}\label{eq:nonlinear_gaussian_filter}
\begin{align}
\frac{\dif \bar{x}(t)}{\dif t} &= \E[\mu(t,X(t))\mid \mathscr{Y}(t)],  \\
\begin{split}
\frac{\dif \Sigma(t)}{\dif t}  &= \bar{Q}_i(t) + \C[\mu(t,X(t)),X(t)\mid \mathscr{Y}(t)]\\
&\quad+ \C[\mu(t,X(t)),X(t)\mid \mathscr{Y}(t)]^\T ,
\end{split}
\end{align}
\end{subequations}
which is the prediction equations for the Gaussian filter, the filter update is handled by conventional discrete-time methods, see \cite{Sarkka2013,Bell1993,Garcia2015,Tronarp2018a}. Note that \Cref{eq:nonlinear_gaussian_filter} gives rise to two kinds of Gaussian filters depending on whether \Cref{alg:slrsde} or \Cref{alg:slrsde2} are used. These shall be referred to as the Gaussian filter of the first kind and the Gaussian filter of the second kind for $i=1$ and $i = 2$, respectively, with $i=1$ giving the classical Gaussian smoother (c.f \cite{SarkkaSarmavuori2013}). This terminology will be used throughout the paper to distinguish between methods derived from \Cref{alg:slrsde,alg:slrsde2}, respectively.  

Since the filtering distribution and the smoothing distribution are equal at $t = t_K$, there are two notable options for deriving approximations to the smoothing distribution, namely (i) linearising with respect to the filtering distribution or (ii) linearising with respect to the smoothing distribution on the fly. In the latter case, the linearisation parameters are given by 
\begin{subequations}\label{eq:smoother_linearisation}
\begin{align}
A(t) &= \C[\mu(t,X(t)),X(t)\mid \mathscr{Y}(t_K)]\Omega^{-1}(t), \\
b(t) &= \E[\mu(t,X(t))\mid \mathscr{Y}(t_K)] - A(t)\hat{x}(t), \\
\bar{\sigma}_i(t) &=
\begin{cases}
\E[\sigma(t,X(t))\sigma^\T(t,X(t))\mid \mathscr{Y}(t_K)]^{1/2}, \ i = 1, \\
\E[\sigma(t,X(t))\mid \mathscr{Y}(t_K)], \ i = 2.
\end{cases}
\end{align}
\end{subequations}
Plugging \Cref{eq:smoother_linearisation} into the linear smoothing equations \Cref{eq:linear_smoothing_equation} gives
\begin{subequations}\label{eq:type1}
\begin{align}
\begin{split}
\frac{\dif \hat{x}(t) }{\dif t}  &= \E[\mu(t,X(t))\mid \mathscr{Y}(t_K)] \\
&+ \bar{Q}_i(t)\Sigma^{-1}(t)[\hat{x}(t) - \bar{x}(t)] ,
\end{split}\\
\begin{split}
\frac{\dif\Omega(t) }{\dif t}  &= \bar{Q}_i(t) \Sigma^{-1}(t)\Omega(t) + \Omega(t) \Sigma^{-1}(t)\bar{Q}_i(t) - \bar{Q}_i(t)\\
&+ \C[\mu(t,X(t)),X(t) \mid \mathscr{Y}(t_K)] \\
&+ \C[\mu(t,X(t)),X(t) \mid \mathscr{Y}(t_K)]^\T .
\end{split}
\end{align}
\end{subequations}
The smoother in \Cref{eq:type1} are referred to as Type I$^*$ of the first kind and  Type I$^*$ of the second kind for $i=1$ and $i=2$, respectively. This, because their similarity to the Type I smoother of \cite{SarkkaSarmavuori2013}, their connection is elaborated on in \Cref{prop:type1_vs_type1star}.
\begin{proposition}\label{prop:type1_vs_type1star}
Let $X(t)$ and $\{Y(t_k)\}_{k=1}^K$ be governed by the system in \Cref{eq:general_model} and assume $\sigma(t,X(t)) = \sigma(t)$. Then the smoothing equations for the Type I \cite{SarkkaSarmavuori2013} and Type I$^*$ smoothers of the first and second kind agree. 
\end{proposition}
\begin{proof}
First note that $\bar{Q}_1 = \bar{Q}_2$ since the diffusion is state independent. The statement then follows by direct comparison to~\cite[$\mathrm{Equation\ }(27)$]{SarkkaSarmavuori2013}.
\end{proof}
 
When the linearisation is done with respect to the filtering distribution the smoothing moments are retrieved by plugging \Cref{eq:filter_linearisation} into \Cref{eq:linear_smoothing_equation}
\begin{subequations}\label{eq:type2}
\begin{align}
\begin{split}
\frac{\dif \hat{x}(t) }{\dif t}   &=  \C[\mu(t,X(t)),X(t)\mid \mathscr{Y}(t)]\\
&\quad\times \Sigma^{-1}(t) (\hat{x}(t)  - \bar{x}(t) )\\
&\quad+ \bar{Q}_i(t)\Sigma^{-1}(t)(\hat{x}(t) - \bar{x}(t))  \\
&\quad+ \E[\mu(t,X(t))\mid \mathscr{Y}(t)],  
\end{split}\\
\begin{split}
\frac{\dif \Omega(t) }{\dif t}  &= \Big( \C[\mu(t,X(t)),X(t)\mid \mathscr{Y}(t)] + \bar{Q}_i(t)\Big)\\
&\quad\times\Sigma^{-1}(t) \Omega(t) \\
&\quad+ \Omega(t) \Sigma^{-1}(t) \\
&\quad\times \Big( \C[\mu(t,X(t)),X(t)\mid \mathscr{Y}(t)]^\T + \bar{Q}_i(t)\Big)  \\
&\quad- \bar{Q}_i(t),
\end{split}
\end{align}
\end{subequations}
which corresponds to the Type II smoother of \cite{SarkkaSarmavuori2013} when $i=1$ and when $i=2$ another smoother is obtained. These shall, again, be referred to as Type II of the first kind and Type II of the second kind for $i=1$ and $i=2$, respectively. Furthermore, by the same argument as in \cite{SarkkaSarmavuori2013}, the Type II formulation may be converted to a Type III formulation, where only forward-time ODEs need to be solved. This argument is not repeated here, but the result is simply (see \cite{SarkkaSarmavuori2013})
\begin{subequations}\label{eq:type3}
\begin{align}
\frac{\dif \bar{x}(t)}{\dif t}  &= A(t) \bar{x}(t) + b(t), \\
\begin{split}
\frac{\dif \Sigma(t)}{\dif t}  &= A(t) \Sigma(t) + \Sigma(t) A^\T(t) + \bar{Q}_i(t),
\end{split}\\
\frac{\dif H_{l}(t)}{\dif t}  &= H_{l}(t)  A^\T(t), \\
G_{l+1} &= H_{l}(t_{l+1}^-) \Sigma^{-1}(t_{l+1}^-), \\
\begin{split}
\hat{x}(t_l) &= \bar{x}(t_l)  + G_{l+1} \Big( \hat{x}(t_{l+1}) - \bar{x}(t_{l+1}^-) \Big),
\end{split} \\
\begin{split}
\Omega(t_l) &=   G_{l+1} \Big(\Omega(t_{l+1}) - \Sigma(t_{l+1}^-) \Big)G_{l+1}^\T \\
&\quad+ \Sigma(t_l) . 
\end{split}
\end{align}
\end{subequations}

\begin{remark}
The Type II and III smoothers of the first kind are precisely the Type II and Type III smoothers as described by \cite{SarkkaSarmavuori2013}.
\end{remark}

To conclude this section we note that \Cref{prop:K1vK2} indicates that the diffusion term in the smoothers of the first kind will be larger than the diffusion term for the smoothers of the second kind. This leads to the expectation that the smoothers of the first kind will report a larger uncertainty than those of the second kind.

\section{Continuous-Discrete Iterative Gaussian Smoothers}\label{sec:iterative_smoothers}
In this section, the linearisation techniques of \Cref{sec:slr_sde} are combined with the Type I$^*$/II/III smoothers (of the first and second kind) of \Cref{sec:smoothers} to develop iterative Gaussian smoothers. This is done in an analogous manner to the discrete-time iterative smoothers \cite{Garcia2017,Tronarp2018a,Bell1994}. The basic idea is that given a Gaussian process, $\widehat{X}(t)$, \Cref{alg:slrsde} or \Cref{alg:slrsde2} can readily be applied to the system in \Cref{eq:dynamics,eq:measurements} (using standard statistical linear regression for the measurement equation \cite{Garcia2015}), which yields an approximate affine system for which inference is straight-forward. An iterative scheme is then obtained by alternating between linearisation and Gaussian smoothing, where $\widehat{X}(t)$ is always chosen as the current best approximation to the smoothing process. This defines an iterative scheme reminiscent of the Gauss--Newton method \cite{Bell1994}. 

\subsection{Iterative Smoothers}
Let $\{\widehat{X}^{(j)}(s)\}_{s\geq 0}$ be a Gaussian process approximating the smoothed process at iteration $j$, with moment functions, $\hat{x}^{(j)}(t)$ and $\Omega^{(j)}(t)$. Moreover, for a test function, $\psi(X)$, denote the expectation of $\psi(\widehat{X}^{(j)}(t))$ by $\E^{(j)}[\psi(\widehat{X}(t))]$. The linearisation parameters, $A^{(j)}$, $b^{(j)}$, and $\bar{\sigma}_i^{(j)}(t)$ can then be obtained by either using \Cref{alg:slrsde} or \Cref{alg:slrsde2} and the linearisation of the measurement model is given by~\cite{Garcia2017,Tronarp2018a,Garcia2015}  
\begin{subequations}
\begin{align*}
C^{(j)}(t_k) &= \C^{(j)}[h(t_k,\widehat{X}(t_k)),\widehat{X}(t_k)](\Omega^{(j)}(t_k))^{-1}, \\
d^{(j)}(t_k) &= \E^{(j)}[h(t_k,\widehat{X}(t_k))] - C^{(j)}(t_k) \hat{x}^{(j)}(t_k),\\
\begin{split}
\Delta^{(j)}(t_k) &= \V^{(j)}[h(t_k,X(t))] + R, \\
&- C^{(j)}(t_k) \Omega^{(j)}(t_k) (C^{(j)}(t_k))^\T,
\end{split}
\end{align*}
\end{subequations}
where $\Delta^{(j)}(t_k)$ is the variance of the residual at iteration $j$. The approximate smoothed process at iteration $j+1$ is then obtained by considering the system:
\begin{subequations}
\begin{align}
\begin{split}
\dif X(t) &= \Big( A^{(j)}(t) X(t) + b^{(j)}(t)  \Big) \dif t \\
&+ \bar{\sigma}^{(j)}(t) \dif \widehat{W}^{(j)}(t),
\end{split} \\
Y({t_k}) &= C^{(j)}(t_k) X(t_k) + d^{(j)}(t_k) + \widehat{V}^{(j)}(t_k), \\ 
\begin{split}
&\C[\widehat{V}^{(j)}(t_k),\widehat{V}^{(j)}(t_l)] =\delta_{k,l}\Delta^{(j)}(t_k), \\
&\widehat{V}^{(j)}(t_k) \sim \mathcal{N}(0,\Delta^{(j)}(t_k)),
\end{split}
\end{align}
\end{subequations}
where $\widehat{W}^{(j)}(t)$ is a standard Wiener process. An approximation to the filtered process at iteration $j+1$, $\{\bar{X}^{(j+1)}(s)\}_{s \geq 0 }$, is then obtained by using the linear filter defined by \Cref{eq:kalman_prediction,eq:kalman_update}, after which any of the smoother formulations \Cref{eq:type1,eq:type2,eq:type3} may be used to obtain $\{\widehat{X}^{(j+1)}(s)\}_{s \geq 0 }$.

\begin{remark}
In practice, \Cref{alg:slrsde,alg:slrsde2} can not be implemented in closed form. Standard approaches to approximate expectations with respect to a Gaussian density is by first order Taylor series or sigma-points \cite{Sarkka2013}. If the first order Taylor series method is used together with \Cref{alg:slrsde} or \Cref{alg:slrsde2}  then a continuous-time iterated extended Kalman smoother is obtained. 
\end{remark}

\subsection{Fixed Point Characterisation}
A convergence analysis of the proposed iteration scheme is beyond the scope of this paper. However, for the smoothers of the first kind, one can discretise the system, apply the analysis of the discrete time case \cite{Garcia2017,Tronarp2018a}, and assume the limits $J\to \infty $ and $\delta t \to 0$ can be interchanged, in which case convergence is guaranteed if the iterations are initialised sufficiently close to a fix point. 

Another topic of investigation is the relationship between the different types of smoothers at the fixed point. More specifically, the relationship between the Type II and Type I$^*$ smoother is illuminated. The smoothing moments for the Type II smoother at iteration $j+1$ are given by  
\begin{subequations}
\begin{align*}
\begin{split}
&\frac{\dif \hat{x}^{(j+1)}(t) }{\dif t}   =  \E^{(j)}[\mu(t,X(t))\mid \mathscr{Y}(t)] \\
&\quad+\bar{Q}_i^{(j)}(t)[\Sigma^{(j+1)}(t)]^{-1}(\hat{x}^{(j+1)}(t) - \bar{x}^{(j+1)}(t))   \\
&\quad+\C^{(j)}[\mu(t,X(t)),X(t)](\Sigma^{(j+1)}(t))^{-1} \\
&\quad\times (\hat{x}^{(j+1)}(t) - \bar{x}^{(j+1)}(t)) 
\end{split}\\
\begin{split}
&\frac{\dif \Omega^{(j+1)}(t)}{\dif t}  = - \bar{Q}_i^{(j)}(t)\\
&\quad+\Big( \C^{(j)}[\mu(t,X(t)),X(t)] + \bar{Q}_i^{(j)}(t)\Big) \\
&\quad \times [\Sigma^{(j+1)}(t)]^{-1} \Omega^{(j+1)}(t) \\
&\quad+ \Omega^{(j+1)}(t) [\Sigma^{(j+1)}(t)]^{-1} \\
&\quad \times \Big( \C^{(j)}[\mu(t,X(t)),X(t)]^\T + \bar{Q}_i(t)\Big)  .
\end{split}
\end{align*}
\end{subequations}
\begin{proposition}
The Type I$^*$, Type II, and Type III smoothers of first and second kinds are equivalent at the fixed point, respectively. That is, they converge to the same point. 
\end{proposition}
\begin{proof}
Assume $(\hat{x}^{(j)},\Omega^{(j)})$ is a fixed point of the iteration and iterate once again. That is, insert $A^{(j)}$, $b^{(j)}$ and $\bar{Q}_i^{(j)}$ into \Cref{eq:type2} to obtain
\begin{subequations}
\begin{align*}
\begin{split}
&\frac{\dif \hat{x}^{(j+1)}(t)}{\dif t}  = \E^{(j)}[\mu(t,X(t))] \\
&\quad+ \bar{Q}_i^{(j)}(t) [\Sigma^{(j+1)}(t)]^{-1}( \hat{x}^{(j+1)}(t) - \bar{x}^{(j+1)}(t))\\
&\quad+ A^{(j)}(t)( \hat{x}^{(j+1)}(t) - \hat{x}^{(j)}(t) - \bar{x}^{(j+1)}(t) + \bar{x}^{(j)}(t) ),
\end{split}\\
\begin{split}
&\frac{\dif \Omega^{(j+1)}(t)}{\dif t}  = A^{(j)}(t) \Omega^{(j+1)}(t) + \Omega^{(j+1)}(t)(A^{(j)}(t))^\T \\
&\quad+ \bar{Q}_i^{(j)}(t) [\Sigma^{(j+1)}(t)]^{-1} \Omega^{(j+1)}(t)  \\
&\quad+ \Omega^{(j+1)}(t)[\Sigma^{(j+1)}(t)]^{-1}\bar{Q}_i^{(j)}(t)   - \bar{Q}_i^{(j)}(t).
\end{split}
\end{align*}
\end{subequations}
Now, plugging in the definition of $A^{(j)}$ and using the fact that $\bar{x}^{(j+1)} = \bar{x}^{(j)}$, $\Sigma^{(j+1)} = \Sigma^{(j)}$, $\hat{x}^{(j+1)} = \hat{x}^{(j)}$, and $\Omega^{(j+1)} = \Omega^{(j)}$, since $(\hat{x}^{(j)},\Omega^{(j)})$ is a fixed point, gives the following 
\begin{subequations}
\begin{align*}
\begin{split}
&\frac{\dif \hat{x}^{(j)}(t)}{\dif t}  = \E^{(j)}[\mu(t,X(t))] \\
&+ \bar{Q}_i^{(j)}(t) [\Sigma^{(j)}(t)]^{-1}( \hat{x}^{(j)}(t) - \bar{x}^{(j)}(t)),
\end{split}\\
\begin{split}
&\frac{\dif \Omega^{(j)}(t)}{\dif t}  = - \bar{Q}_i^{(j)}(t) + \C^{(j)}[\mu(t,X(t)),X(t)] \\
&\quad+  \C^{(j)}[\mu(t,X(t)),X(t)]^\T \\
&\quad+ \bar{Q}_i^{(j)}(t) [\Sigma^{(j)}(t)]^{-1} \Omega^{(j)}(t)  \\
&\quad+ \Omega^{(j)}(t)[\Sigma^{(j)}(t)]^{-1}\bar{Q}_i^{(j)}(t),
\end{split}
\end{align*}
\end{subequations}
which is the differential equations satisfied by a Type I$^*$ smoother (see \Cref{eq:type1}). Since Type III is equivalent to Type II,  all the presented smoothers (of the same kind) satisfy the same differential equation at the fixed point. 
\end{proof}
\subsection{Computational Complexity and Storage Requirement}
It is important to consider the computational complexity and storage requirement of the different types of iterative smoothers. If the time interval, for purposes of numerical solving the ODEs, is sub-divided into $N$ time stamps and $K$ measurements are processed, then for the non-iterative smoothers it was found that Type III is superior to Type I and II in terms of storage requirement, while being comparable in the number of Gaussian integrals needed \cite{SarkkaSarmavuori2013}. 

However, for the iterative schemes the storage requirements for Type I$^*$ and Type II smoothers are doubled due to having to store the smoothing solution of the previous iteration. The change for Type III smoother is more dramatic since the linearisation requires the storage of the smoothing solution of the previous iteration at all of the $N$ time stamp. The computational requirements for the smoothers using $J$ iterations are given in \Cref{tab:complexity}.  
\begin{table}[t]
	\centering
	\caption{Computational requirements for the iterative smoothers (of any kind).}
	\label{tab:complexity}
	\begin{tabular}{|l c  c|}
	\hline Smoothers & Integrals & Storage      \\
	\hline Type I$^*$   & 10NKJ &  $2NK(d_X+d_X^2)$   \\
	\hline Type II  & 3NKJ  &  $2NK(2d_X+3d_X^2)$ \\
	\hline Type III  & 3NKJ  &  $2NK(2d_X+3d_X^2)$ \\ \hline 
	\end{tabular} 
\end{table}
Therefore there is no significant difference in computational requirements once iterations are introduced.

\section{Experimental Results}\label{sec:experimental_results}

\subsection{Reentry}
The proposed iterative Gaussian smoother is compared to the variational smoother of \cite{Ala-Luhtala2015} in a reentry tracking problem. The state, $U = [X, Y, \dot{X}, \dot{Y}, \Psi]^\T$, represents the position $(X,Y)$, velocity $(\dot{X},\dot{Y})$, and an aerodynamic parameter, $\Psi$ of a vehicle. The dynamic equation is given by 
\begin{equation*}
\begin{split}
\dif U(t) &= \begin{bmatrix} 0 & \mathrm{I}_2 & 0 \\ G(t,U(t))\mathrm{I}_2 & D(t,U(t))\mathrm{I}_2 & 0 \\ 0 & 0 & 0 \end{bmatrix}U(t)\dif t\\
&\quad+ \begin{bmatrix} 0& \mathrm{I}_3\end{bmatrix}^\T\sigma\dif W(t),
\end{split}
\end{equation*}
where $\mathrm{I}_p$ is a $p\times p$ identity matrix and the zero entries are zero matrices of appropriate sizes. The functions $G(t,u)$ and $D(t,u)$ are given by
\begin{subequations}
\begin{align*}
G(t,u) &= - \frac{Gm_0}{(x^2 + y^2)^{3/2}},\\
\begin{split}
D(t,u) &= - \beta_0 \exp \Big(\psi + \frac{R_0 - (x^2 + y^2)^{1/2}}{H_0} \Big)\\
&\quad\times (\dot{x}^2 + \dot{y}^2)^{1/2}.
\end{split}
\end{align*}
\end{subequations}
The parameters were set to  
\begin{equation*}
\sigma = \diag[ \sqrt{2.4064}\cdot 10^{-5/2},\ \sqrt{2.4064} \cdot 10^{-5/2},\  1\cdot 10^{-3}],
\end{equation*}
$\beta_0 = -0.59783$, $H_0 = 13.406$, $Gm_0 = 3.9860 \cdot 10^5$, and $R_0 = 6374$. The vehicle is measured once per second by a radar at position $[s_x,s_y]^\T$ according to 
\begin{equation*}
\begin{split}
Z(t_k) &= \begin{bmatrix}  [(X(t_k) - s_x)^2 + (Y(t_k) - s_y)^2 ]^{1/2} \\ \tan^{-1} \Big( \frac{Y(t_k) - s_y}{X(t_k) - s_x} \Big) \end{bmatrix} \\
&\quad + V(t_k), 
\end{split}
\end{equation*}
where $V(t_k)$ is a Gaussian white noise sequence with covariance matrix
\begin{equation*}
R = \diag[1\cdot 10^{-3},\ 1.7\cdot 10^{-3}]. 
\end{equation*}
The initial state, $X(0)$ is Gaussian distributed with moments 
\begin{subequations}
\begin{align*}
\bar{x}(0) &= \big[ 6500.4 ,\ 349.14 ,\ -1.8093 ,\ -6.7967 ,\ 0.6932 \big],\\
\Sigma(0) &= \begin{bmatrix} \mathrm{I}_4 \cdot 10^{-6} & 0 \\ 0 & 1 \end{bmatrix}.
\end{align*}
\end{subequations}
The system was simulated 100 times on the interval $t\in [0,200]$ using the Euler-Maruyama method with a step-size of $1/1000$. The Type III template (\Cref{eq:type3}) is used for the implementation of the proposed iterative Gaussian smoother,\footnote{Note that since the diffusion is state independent the iterative smoothers of the first and second kind are equivalent, see \Cref{prop:type1_vs_type1star}.} the initial linearisation being with respect to the filtering distributions, with up to 4 subsequent iterations. The ODEs are approximated by constant input between discretisation instants, that is zeroth order hold whereby the equivalent discrete time system is computed using the matrix fraction decomposition (see, e.g, \cite{Axelsson2015}). The performance is compared to the variational smoother of \cite{Ala-Luhtala2015} \footnote{Since the diffusion is constant, the smoothers of \cite{Ala-Luhtala2015} and \cite{Sutter2016} are equivalent.}, which uses the standard fourth order Runge-Kutta method for integration, the same expectation approximator, and iterates until the change in Kullback-Leibler divergence is less than $10^{-3}$, the adaptive step-size goes below the threshold $10^{-3}$, or 20 iterations have been performed. Both smoothers use a step-size of $\delta t = 1/100$ for time integration and the spherical-radial cubature rule \cite{Arasaratnam2009} to approximate expectations. 

Boxplots for the RMSE in position, velocity, and the aerodynamic parameter are shown in \Cref{fig:reentry_rmse_boxplot} for iterations 0 through 4 of the proposed smoother and the variational smoother at convergence. The proposed smoother converges after a few iterations have similar performance to the variational smoother. 

The $\chi^2$-statistic (also known as NEES \cite{BarShalom2004}), averaged over Monte Carlo trials, for the various smoothers is shown in \Cref{fig:reentry_chi2stats}. Clearly the initialisation is inconsistent. Interestingly, the iterations appear to fall below the lower confidence band, which may be indicative of an overestimated covariance. In contrast, the variational smoother tends to have a slightly larger $\chi^2$-statistic on average, while still mostly keeping itself within the confidence band. 

Lastly, averaged performance metrics are reported in \Cref{tab:reentry_summary}. Here it can be seen that the proposed smoother again converges rapidly, after 1-2 iterations. The variational smoother took an average of $6.5$ iterations to converge, though the convergence criterion was rather strict so it can not be excluded that it can use fewer iterations without making a significant sacrifice in performance.    

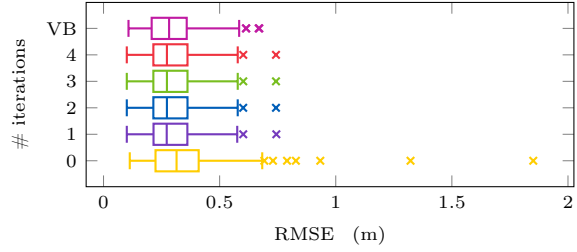
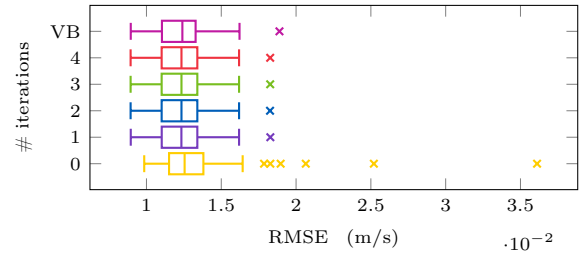
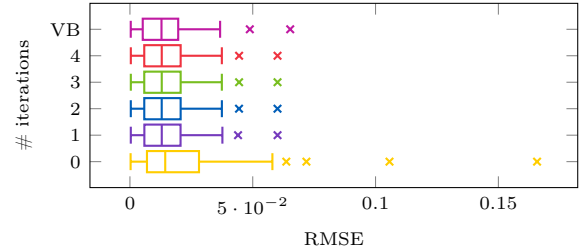
\begin{figure}[t!]
\centering
\subfloat[\scriptsize{Position}]{\begin{tikzpicture}
\begin{axis}[
       width=\columnwidth,
       height=0.5\columnwidth,
      ytick={1,2,3,4,5,6},
      yticklabels={$0$,$1$,$2$,$3$,$4$,$\text{VB}$},
       xlabel={RMSE \, $(\mathrm{m})$},
       ylabel={\# iterations},
      label style ={font=\scriptsize},
      ticklabel style = {font=\scriptsize},
]
\addplot[AaltoYellow,mark=x,thick,boxplot] 
      table[row sep=newline,y index=0]{../fig/REENTRY_BOXPLOT_RMSE_POS1.txt};

\addplot[AaltoPurple,mark=x,thick,boxplot] 
      table[row sep=newline,y index=0]{../fig/REENTRY_BOXPLOT_RMSE_POS2.txt};

\addplot[AaltoBlue,mark=x,thick,boxplot] 
     table[row sep=newline,y index=0]{../fig/REENTRY_BOXPLOT_RMSE_POS3.txt};

\addplot[AaltoGreen,mark=x,thick,boxplot] 
      table[row sep=newline,y index=0]{../fig/REENTRY_BOXPLOT_RMSE_POS4.txt};

\addplot[AaltoRed,mark=x,thick,boxplot] 
      table[row sep=newline,y index=0]{../fig/REENTRY_BOXPLOT_RMSE_POS5.txt};

\addplot[AaltoPurple2,mark=x,thick,boxplot] 
      table[row sep=newline,y index=0]{../fig/REENTRY_BOXPLOT_RMSE_POSVB.txt};

\end{axis}
\end{tikzpicture}}\\
\subfloat[\scriptsize{Velocity}]{\begin{tikzpicture}
\begin{axis}[
       width=\columnwidth,
       height=0.5\columnwidth,
      ytick={1,2,3,4,5,6},
      yticklabels={$0$,$1$,$2$,$3$,$4$,$\text{VB}$},
       xlabel={RMSE \, $(\mathrm{m/s})$},
       ylabel={\# iterations},
      label style ={font=\scriptsize},
      ticklabel style = {font=\scriptsize},
]
\addplot[AaltoYellow,mark=x,thick,boxplot] 
      table[row sep=newline,y index=0]{../fig/REENTRY_BOXPLOT_RMSE_VEL1.txt};

\addplot[AaltoPurple,mark=x,thick,boxplot] 
      table[row sep=newline,y index=0]{../fig/REENTRY_BOXPLOT_RMSE_VEL2.txt};

\addplot[AaltoBlue,mark=x,thick,boxplot] 
     table[row sep=newline,y index=0]{../fig/REENTRY_BOXPLOT_RMSE_VEL3.txt};

\addplot[AaltoGreen,mark=x,thick,boxplot] 
      table[row sep=newline,y index=0]{../fig/REENTRY_BOXPLOT_RMSE_VEL4.txt};

\addplot[AaltoRed,mark=x,thick,boxplot] 
      table[row sep=newline,y index=0]{../fig/REENTRY_BOXPLOT_RMSE_VEL5.txt};

\addplot[AaltoPurple2,mark=x,thick,boxplot] 
      table[row sep=newline,y index=0]{../fig/REENTRY_BOXPLOT_RMSE_VELVB.txt};

\end{axis}
\end{tikzpicture}}\\
\subfloat[\scriptsize{Aerodynamic parameter}]{\begin{tikzpicture}
\begin{axis}[
       width=\columnwidth,
       height=0.5\columnwidth,
      ytick={1,2,3,4,5,6},
      yticklabels={$0$,$1$,$2$,$3$,$4$,$\text{VB}$},
       xlabel={RMSE},
       ylabel={\# iterations},
      label style ={font=\scriptsize},
      ticklabel style = {font=\scriptsize},
]
\addplot[AaltoYellow,mark=x,thick,boxplot] 
      table[row sep=newline,y index=0]{../fig/REENTRY_BOXPLOT_RMSE_PAR1.txt};

\addplot[AaltoPurple,mark=x,thick,boxplot] 
      table[row sep=newline,y index=0]{../fig/REENTRY_BOXPLOT_RMSE_PAR2.txt};

\addplot[AaltoBlue,mark=x,thick,boxplot] 
     table[row sep=newline,y index=0]{../fig/REENTRY_BOXPLOT_RMSE_PAR3.txt};

\addplot[AaltoGreen,mark=x,thick,boxplot] 
      table[row sep=newline,y index=0]{../fig/REENTRY_BOXPLOT_RMSE_PAR4.txt};

\addplot[AaltoRed,mark=x,thick,boxplot] 
      table[row sep=newline,y index=0]{../fig/REENTRY_BOXPLOT_RMSE_PAR5.txt};

\addplot[AaltoPurple2,mark=x,thick,boxplot] 
      table[row sep=newline,y index=0]{../fig/REENTRY_BOXPLOT_RMSE_PARVB.txt};

\end{axis}
\end{tikzpicture}}
\caption{\scriptsize{Boxplots of the RMSE distributions over the 100 Monte Carlo trajectories for position (Top), velocity (middle) and, aerodynamic parameter (bottom), for the variational smoother at convergence (VB) and iterations 0 through 4 of the proposed iterative Gaussian smoother (0-4).}}\label{fig:reentry_rmse_boxplot}
\end{figure}

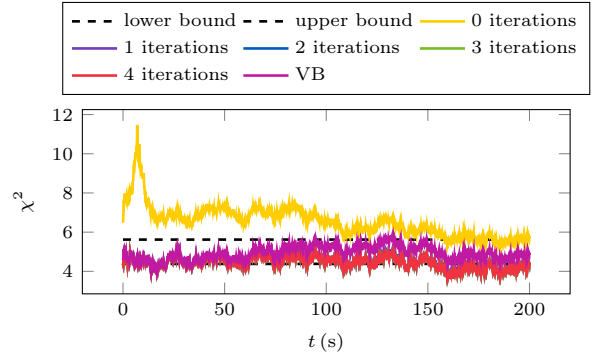
\begin{figure}[t]
\centering
\begin{tikzpicture}
    \begin{axis}[
       width=\columnwidth,
       height=0.5\columnwidth,
	    xlabel	= {$t \, (\mathrm{s})$},
            ylabel = {$\chi^2$},
      ticklabel style = {font=\scriptsize},
      label style ={font=\scriptsize},
	legend cell align = left,
        legend style = {
            font = \scriptsize, 
	    legend columns = 3,
	   at={(axis cs:100,13)},
	    anchor=south
        },
    ]

    \addplot[
        mark = none,
        draw = black,
        dashed,
        line width = 1pt,
    ]
    table[
	    x = t,
	    y = lb,
    ]{../fig/REENTRY_CHI2STATS.txt};
    \addlegendentry{lower bound}

    \addplot[
        mark = none,
        draw = black,
        dashed,
        line width = 1pt,
    ]
    table[
	    x = t,
	    y = ub,
    ]{../fig/REENTRY_CHI2STATS.txt};
    \addlegendentry{upper bound}

    \addplot[
        mark = none,
        draw = AaltoYellow,
        solid,
        line width = 1pt,
    ]
    table[
	    x = t,
	    y = 0,
    ]{../fig/REENTRY_CHI2STATS.txt};
    \addlegendentry{0 iterations}

    \addplot[
        mark = none,
        draw = AaltoPurple,
        solid,
        line width = 1pt,
    ]
    table[
	    x = t,
	    y = 1,
    ]{../fig/REENTRY_CHI2STATS.txt};
    \addlegendentry{1 iterations}

    \addplot[
        mark = none,
        draw = AaltoBlue,
        solid,
        line width = 1pt,
    ]
    table[
	    x = t,
	    y = 2,
    ]{../fig/REENTRY_CHI2STATS.txt};
    \addlegendentry{2 iterations}

    \addplot[
        mark = none,
        draw = AaltoGreen,
        solid,
        line width = 1pt,
    ]
    table[
	    x = t,
	    y = 3,
    ]{../fig/REENTRY_CHI2STATS.txt};
    \addlegendentry{3 iterations}

    \addplot[
        mark = none,
        draw = AaltoRed,
        solid,
        line width = 1pt,
    ]
    table[
	    x = t,
	    y = 4,
    ]{../fig/REENTRY_CHI2STATS.txt};
    \addlegendentry{4 iterations}

    \addplot[
        mark = none,
        draw = AaltoPurple2,
        solid,
        line width = 1pt,
    ]
    table[
	    x = t,
	    y = VB,
    ]{../fig/REENTRY_CHI2STATS.txt};
    \addlegendentry{VB}

    \end{axis}
\end{tikzpicture}
\caption{\scriptsize{The $\chi^2$-statistic at each time point, averaged over the Monte Carlo trials with $95\%$ confidence bands for the variational smoother (VB) and iterations 0 through 4 of the proposed iterative Gaussian smoother (0-4).}}\label{fig:reentry_chi2stats}
\end{figure}

\begin{table}[t]
	\centering
	\caption{\scriptsize{The RMSE in position (POS) $(\mathrm{m})$, velocity (VEL) $(\mathrm{m/s})$, aerodynamic parameter ($\Psi$)  averaged over the Monte Carlo trials, and the average $\chi^2$-statistic, for the variational smoother at convergence (VB) and iterations 0 through 4 of the proposed iterative Gaussian smoother.}}
	\label{tab:reentry_summary}
	\begin{tabular}{|l |c c c c |}
	\hline Method & POS & VEL & $\Psi$ & $\chi^2$        \\
	\hline 0   & 0.3651  & 0.0132  & 0.0208   & 6.5556   \\ 
	\hline 1   & 0.2968  & 0.0123  & 0.0138   & 4.5173   \\ 
	\hline 2   & 0.2967  & 0.0123  & 0.0138   & 4.4565   \\ 
	\hline 3   & 0.2967  & 0.0123  & 0.0138   & 4.4565   \\ 
	\hline 4   & 0.2967  & 0.0123  & 0.0138   & 4.4565   \\ \hline
	\hline VB  & 0.2988  & 0.0124  & 0.0142   & 4.9332   \\ \hline
	\end{tabular} 
\end{table}

\subsection{Radar Tracked Coordinated Turn}
The proposed iterative smoothers are assessed in the radar tracked three dimensional coordinated turn model with state dependent diffusion (see \cite{SarkkaSarmavuori2013}). The latent process, $U = (X,Y,Z,\dot{X},\dot{Y},\dot{Z},\Psi)$, is given by
\begin{equation}\label{eq:radar_dynamics}
\begin{split}
&\mu^\T(u) = \big[ \dot{x},\  \dot{y},\ \dot{z},\ -\psi\dot{y},\ \psi\dot{x},\ 0,\ 0  \big],\\
&\mathrm{d}U(t) = \mu(U(t))\dif t + \sigma(U(t)) \dif W(t),
\end{split}
\end{equation}
where $(X,Y,Z)$ are the position coordinates, $(\dot{X},\dot{Y},\dot{Z})$ the corresponding velocities, $\Psi$ is the turn rate, $W(t)$ is a 4-dimensional Brownian motion, and 
\begin{equation*}
\begin{split}
\sigma(u) =& \begin{bmatrix}
0 & 0 & 0 & 0 \\
0 & 0 & 0 & 0 \\
0 & 0 & 0 & 0 \\
\frac{\dot{x}}{\xi(u)} & \frac{\dot{y}}{\eta(u)} & \frac{\dot{x}\dot{z}}{\xi(u)\eta(u)}  & 0 \\
\frac{\dot{y}}{\xi(u)} & -\frac{\dot{x}}{\eta(u)} & \frac{\dot{y}\dot{z}}{\xi(u)\eta(u)} & 0\\
\frac{\dot{z}}{\xi(u)} & 0 & -\frac{ \eta(u)}{\xi(u)} & 0 \\
0 & 0 & 0 & 1
\end{bmatrix}\\
&\quad\times \begin{bmatrix} \sigma_\parallel & 0 & 0 & 0 \\ 0 & \sigma_h & 0 & 0 \\ 0 & 0 & \sigma_v & 0 \\ 0 &0&0&\sigma_\Psi \end{bmatrix},
\end{split}
\end{equation*}
where
\begin{equation*}
\xi(u) = \sqrt{x^2+y^2+z^2}, \ \eta(u) = \sqrt{\dot{x}^2 + \dot{y}^2}.
\end{equation*}
The system is measured according to 
\begin{equation}\label{eq:radar_measurement}
\begin{split}
\begin{bmatrix} \Lambda(t_k) \\ \Theta(t_k) \\ \Phi(t_k) \end{bmatrix} =& \begin{bmatrix} \sqrt{X^2(t_k) + Y^2(t_k) + Z^2(t_k)} \\
 \tan^{-1} [Y(t_k)/X(t_k)]\\
  \tan^{-1} \frac{Z(t_k)}{\sqrt{X^2(t_k)+Y^2(t_k)}}   \end{bmatrix}
 + V(t_k),\\
&V(t_k) \sim \mathcal{N}\Big(0,\diag \begin{bmatrix} \sigma_\Lambda^2 & \sigma_\Theta^2 & \sigma_\Phi^2 \end{bmatrix}\Big). 
\end{split}
\end{equation}
The parameters were set as follows, $\sigma_\parallel = \sqrt{100}$, $\sigma_h = \sqrt{0.2}$,  $\sigma_v = \sqrt{0.2}$, $\sigma_\Psi = 7\cdot 10^{-3}\mathrm{rad}/\mathrm{s}$ , $\sigma_\Lambda = 50\mathrm{m}$ , $\sigma_\Theta = \sigma_\Phi = 0.1\pi/180 \mathrm{rad}$. The statistics of the initial state was set to $U(0) \sim \mathcal{N}(\bar{x}(0^-),\Sigma(0^-))$, where 
\begin{subequations}
\begin{align*}
&\bar{x}(0^-) = \begin{bmatrix} 1000 & 0 & 2650 & 200 & 0 & 150 & \frac{6\pi}{180} \end{bmatrix}^\T, \\
&\Sigma(0^-) = 100^2\diag\Big[ 1,\ 1,\ 1,\  1,\ 1,\ 1,\ \frac{\pi}{180\cdot 100^2}\Big].
\end{align*}
\end{subequations}
It should be noted that the diffusion term in \Cref{eq:radar_dynamics} is both singular and state-dependent, hence there exists no Gaussian process with respect to which the probability law of $U(t)$ is absolutely continuous and consequently none of variational smoothers by~\cite{Archambeau2007,Ala-Luhtala2015} are applicable. Moreover, as mentioned, the method of \cite{Sutter2016} requires solving a 7-dimensional stochastic partial differential equation, which is computationally unattractive. Similarly, for expectation propagation \cite{Cseke2016} it is not clear how to form good approximations of the expectation with respect to the tilted distribution with a state dimension of 7 and likelihoods in \Cref{eq:radar_measurement}. However, this provides a good opportunity to compare the iterative smoothers of the first and second kind (K1 and K2, respectively.). 

The Euler-Maruyama method was used to generate 100 independent realisations of the system using a step-size of $5/1000$, with the time between measurements set to $\Delta T = 6$, with $26$ measurement instants in total, starting from $t = 0$. Both smoothers were implemented in the same manner as the previous experiment, using a step-size of $\delta t = 5/100$.   

Boxplots of the root-mean-square error (RMSE) distribution over the 100 Monte-Carlo trajectories is provided in \Cref{fig:coordinated_turn_rmse_boxplot} for position, velocity, and turn-rate, respectively. It is clear that iterations can offer a substantial improvement in accuracy. The consistency of the iterative smoothers is assessed using the $ \chi^2$-statistic, which is averaged over the Monte Carlo trajectories and the resulting time series for iterations 0 through 4 is provided in \Cref{fig:coordinated_turn_chi2stats}. It can be seen that the initialisation of the smoothers is grossly inconsistent and the first iteration provides a massive improvement, while the subsequent iterations provide smaller improvements. Furthermore, the average RMSE and the average $\chi^2$ statistics for the different iterations is shown in \Cref{tab:radar_tracking_summary}.

The impression is that the smoother converges rather quickly, after two to three iterations in this scenario. Also the iterative smoothers K1 and K2 appear to perform similarly on this problem, with no discernible difference on average for up to 4 significant digits at convergence. However, the iterative smoother of the second kind perform notably worse at initialisation, particularly in terms of consistency, see  \Cref{tab:radar_tracking_summary}. Drawing from the discussion in \Cref{subsec:dtl_vs_vf}, we know that both smoothers will be approximately equivalent when the variance in the linearising process is small. As the system noises are fairly small in this experiment this eventually becomes the case as the iterations proceed.   

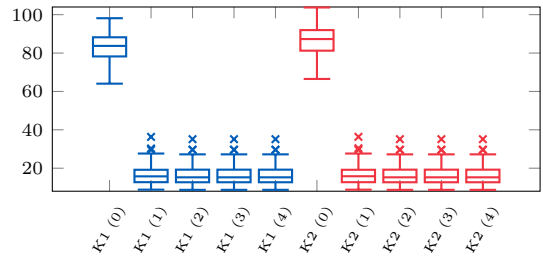
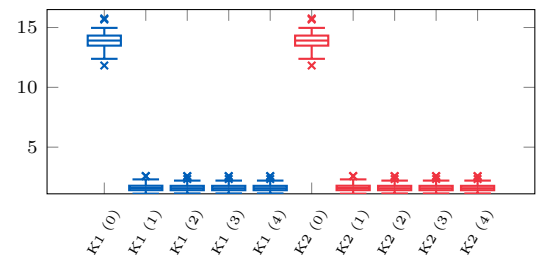
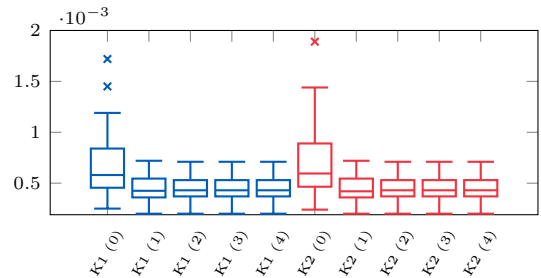
\begin{figure}[t!]
\centering
\subfloat[\scriptsize{Position}]{\begin{tikzpicture}
\begin{axis}[
       boxplot/draw direction=y,
       width=\columnwidth,
       height=0.5\columnwidth,
      xtick={1,2,3,4,5,6,7,8,9,10},
      xticklabels={K1 (0),K1 (1),K1 (2),K1 (3),K1 (4),K2 (0),K2 (1),K2 (2),K2 (3),K2 (4)},
      ticklabel style = {font=\scriptsize},
     xticklabel style = {rotate=60,font=\tiny},
      label style ={font=\scriptsize},
       ymax=104,
       ymin=8,
]
\addplot[AaltoBlue,mark=x,thick,boxplot] 
      table[row sep=newline,y index=0]{../fig/COORDINATED_TURN_BOXPLOT_RMSE_POS1_K1.txt};

\addplot[AaltoBlue,mark=x,thick,boxplot] 
      table[row sep=newline,y index=0]{../fig/COORDINATED_TURN_BOXPLOT_RMSE_POS2_K1.txt};

\addplot[AaltoBlue,mark=x,thick,boxplot] 
     table[row sep=newline,y index=0]{../fig/COORDINATED_TURN_BOXPLOT_RMSE_POS3_K1.txt};

\addplot[AaltoBlue,mark=x,thick,boxplot] 
      table[row sep=newline,y index=0]{../fig/COORDINATED_TURN_BOXPLOT_RMSE_POS4_K1.txt};

\addplot[AaltoBlue,mark=x,thick,boxplot] 
      table[row sep=newline,y index=0]{../fig/COORDINATED_TURN_BOXPLOT_RMSE_POS5_K1.txt};

\addplot[AaltoRed,mark=x,thick,boxplot] 
      table[row sep=newline,y index=0]{../fig/COORDINATED_TURN_BOXPLOT_RMSE_POS1_K2.txt};

\addplot[AaltoRed,mark=x,thick,boxplot] 
      table[row sep=newline,y index=0]{../fig/COORDINATED_TURN_BOXPLOT_RMSE_POS2_K2.txt};

\addplot[AaltoRed,mark=x,thick,boxplot] 
     table[row sep=newline,y index=0]{../fig/COORDINATED_TURN_BOXPLOT_RMSE_POS3_K2.txt};

\addplot[AaltoRed,mark=x,thick,boxplot] 
      table[row sep=newline,y index=0]{../fig/COORDINATED_TURN_BOXPLOT_RMSE_POS4_K2.txt};

\addplot[AaltoRed,mark=x,thick,boxplot] 
      table[row sep=newline,y index=0]{../fig/COORDINATED_TURN_BOXPLOT_RMSE_POS5_K2.txt};

\end{axis}
\end{tikzpicture}}\\
\subfloat[\scriptsize{Velocity}]{\begin{tikzpicture}
\begin{axis}[
       boxplot/draw direction=y,
       width=\columnwidth,
       height=0.5\columnwidth,
      xtick={1,2,3,4,5,6,7,8,9,10},
      xticklabels={K1 (0),K1 (1),K1 (2),K1 (3),K1 (4),K2 (0),K2 (1),K2 (2),K2 (3),K2 (4)},
      ticklabel style = {font=\scriptsize},
     xticklabel style = {rotate=60,font=\tiny},
      label style ={font=\scriptsize},
       ymax=16.5,
       ymin=1.1,
]
\addplot[AaltoBlue,mark=x,thick,boxplot] 
      table[row sep=newline,y index=0]{../fig/COORDINATED_TURN_BOXPLOT_RMSE_VEL1_K1.txt};

\addplot[AaltoBlue,mark=x,thick,boxplot] 
      table[row sep=newline,y index=0]{../fig/COORDINATED_TURN_BOXPLOT_RMSE_VEL2_K1.txt};

\addplot[AaltoBlue,mark=x,thick,boxplot] 
     table[row sep=newline,y index=0]{../fig/COORDINATED_TURN_BOXPLOT_RMSE_VEL3_K1.txt};

\addplot[AaltoBlue,mark=x,thick,boxplot] 
      table[row sep=newline,y index=0]{../fig/COORDINATED_TURN_BOXPLOT_RMSE_VEL4_K1.txt};

\addplot[AaltoBlue,mark=x,thick,boxplot] 
      table[row sep=newline,y index=0]{../fig/COORDINATED_TURN_BOXPLOT_RMSE_VEL5_K1.txt};

\addplot[AaltoRed,mark=x,thick,boxplot] 
      table[row sep=newline,y index=0]{../fig/COORDINATED_TURN_BOXPLOT_RMSE_VEL1_K2.txt};

\addplot[AaltoRed,mark=x,thick,boxplot] 
      table[row sep=newline,y index=0]{../fig/COORDINATED_TURN_BOXPLOT_RMSE_VEL2_K2.txt};

\addplot[AaltoRed,mark=x,thick,boxplot] 
     table[row sep=newline,y index=0]{../fig/COORDINATED_TURN_BOXPLOT_RMSE_VEL3_K2.txt};

\addplot[AaltoRed,mark=x,thick,boxplot] 
      table[row sep=newline,y index=0]{../fig/COORDINATED_TURN_BOXPLOT_RMSE_VEL4_K2.txt};

\addplot[AaltoRed,mark=x,thick,boxplot] 
      table[row sep=newline,y index=0]{../fig/COORDINATED_TURN_BOXPLOT_RMSE_VEL5_K2.txt};

\end{axis}
\end{tikzpicture}}\\
\subfloat[\scriptsize{Turn-rate}]{\begin{tikzpicture}
\begin{axis}[
       boxplot/draw direction=y,
       width=\columnwidth,
       height=0.5\columnwidth,
      xtick={1,2,3,4,5,6,7,8,9,10},
      xticklabels={K1 (0),K1 (1),K1 (2),K1 (3),K1 (4),K2 (0),K2 (1),K2 (2),K2 (3),K2 (4)},
      ticklabel style = {font=\scriptsize},
     xticklabel style = {rotate=60,font=\tiny},
      label style ={font=\scriptsize},
       ymax=0.002,
       ymin=0.00019,
]
\addplot[AaltoBlue,mark=x,thick,boxplot] 
      table[row sep=newline,y index=0]{../fig/COORDINATED_TURN_BOXPLOT_RMSE_ANG1_K1.txt};

\addplot[AaltoBlue,mark=x,thick,boxplot] 
      table[row sep=newline,y index=0]{../fig/COORDINATED_TURN_BOXPLOT_RMSE_ANG2_K1.txt};

\addplot[AaltoBlue,mark=x,thick,boxplot] 
     table[row sep=newline,y index=0]{../fig/COORDINATED_TURN_BOXPLOT_RMSE_ANG3_K1.txt};

\addplot[AaltoBlue,mark=x,thick,boxplot] 
      table[row sep=newline,y index=0]{../fig/COORDINATED_TURN_BOXPLOT_RMSE_ANG4_K1.txt};

\addplot[AaltoBlue,mark=x,thick,boxplot] 
      table[row sep=newline,y index=0]{../fig/COORDINATED_TURN_BOXPLOT_RMSE_ANG5_K1.txt};

\addplot[AaltoRed,mark=x,thick,boxplot] 
      table[row sep=newline,y index=0]{../fig/COORDINATED_TURN_BOXPLOT_RMSE_ANG1_K2.txt};

\addplot[AaltoRed,mark=x,thick,boxplot] 
      table[row sep=newline,y index=0]{../fig/COORDINATED_TURN_BOXPLOT_RMSE_ANG2_K2.txt};

\addplot[AaltoRed,mark=x,thick,boxplot] 
     table[row sep=newline,y index=0]{../fig/COORDINATED_TURN_BOXPLOT_RMSE_ANG3_K2.txt};

\addplot[AaltoRed,mark=x,thick,boxplot] 
      table[row sep=newline,y index=0]{../fig/COORDINATED_TURN_BOXPLOT_RMSE_ANG4_K2.txt};

\addplot[AaltoRed,mark=x,thick,boxplot] 
      table[row sep=newline,y index=0]{../fig/COORDINATED_TURN_BOXPLOT_RMSE_ANG5_K2.txt};

\end{axis}
\end{tikzpicture}}
\caption{\scriptsize{Boxplots of the RMSE distributions over the 100 Monte Carlo trajectories for position (top), velocity (middle), and turn-rate (bottom), for the iterative Gaussian smoothers K1 and K2. The number of iterations after initialisation is shown in parenthesis.}}\label{fig:coordinated_turn_rmse_boxplot}
\end{figure}

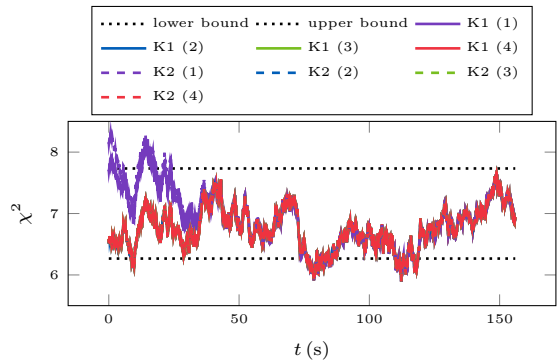
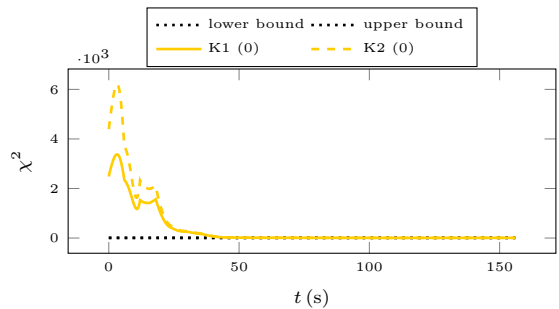
\begin{figure}[t!]
\centering
\subfloat[\scriptsize{$\chi^2$-statistics for iterations 1 through 4 for both kinds.}]{\begin{tikzpicture}
    \begin{axis}[
       width=\columnwidth,
       height=0.5\columnwidth,
	    xlabel	= {$t \, (\mathrm{s})$},
            ylabel = {$\chi^2$},
      ticklabel style = {font=\tiny},
      label style ={font=\scriptsize},
	    ymin = 5.5,
	    ymax = 8.5,
	legend cell align = left,
        legend style = {
            font = \tiny, 
	    legend columns = 3,
	   at={(axis cs:78,8.6)},
	    anchor=south
        },
    ]

    \addplot[
        mark = none,
        draw = black,
        dotted,
        line width = 1pt,
    ]
    table[
	    x = t,
	    y = lb,
    ]{../fig/COORDINATED_TURN_CHI2STATS.txt};
    \addlegendentry{lower bound}

    \addplot[
        mark = none,
        draw = black,
        dotted,
        line width = 1pt,
    ]
    table[
	    x = t,
	    y = ub,
    ]{../fig/COORDINATED_TURN_CHI2STATS.txt};
    \addlegendentry{upper bound}


    \addplot[
        mark = none,
        draw = AaltoPurple,
        solid,
        line width = 1pt,
    ]
    table[
	    x = t,
	    y = K12,
    ]{../fig/COORDINATED_TURN_CHI2STATS.txt};
    \addlegendentry{K1 (1)}

    \addplot[
        mark = none,
        draw = AaltoBlue,
        solid,
        line width = 1pt,
    ]
    table[
	    x = t,
	    y = K13,
    ]{../fig/COORDINATED_TURN_CHI2STATS.txt};
    \addlegendentry{K1 (2)}

    \addplot[
        mark = none,
        draw = AaltoGreen,
        solid,
        line width = 1pt,
    ]
    table[
	    x = t,
	    y = K14,
    ]{../fig/COORDINATED_TURN_CHI2STATS.txt};
    \addlegendentry{K1 (3)}

    \addplot[
        mark = none,
        draw = AaltoRed,
        solid,
        line width = 1pt,
    ]
    table[
	    x = t,
	    y = K15,
    ]{../fig/COORDINATED_TURN_CHI2STATS.txt};
    \addlegendentry{K1 (4)}


    \addplot[
        mark = none,
        draw = AaltoPurple,
        dashed,
        line width = 1pt,
    ]
    table[
	    x = t,
	    y = K22,
    ]{../fig/COORDINATED_TURN_CHI2STATS.txt};
    \addlegendentry{K2 (1)}

    \addplot[
        mark = none,
        draw = AaltoBlue,
        dashed,
        line width = 1pt,
    ]
    table[
	    x = t,
	    y = K23,
    ]{../fig/COORDINATED_TURN_CHI2STATS.txt};
    \addlegendentry{K2 (2)}

    \addplot[
        mark = none,
        draw = AaltoGreen,
        dashed,
        line width = 1pt,
    ]
    table[
	    x = t,
	    y = K24,
    ]{../fig/COORDINATED_TURN_CHI2STATS.txt};
    \addlegendentry{K2 (3)}

    \addplot[
        mark = none,
        draw = AaltoRed,
        dashed,
        line width = 1pt,
    ]
    table[
	    x = t,
	    y = K25,
    ]{../fig/COORDINATED_TURN_CHI2STATS.txt};
    \addlegendentry{K2 (4)}

    \end{axis}
\end{tikzpicture}}\\
\subfloat[\scriptsize{$\chi^2$-statistics for initialisation for both kinds.}]{\begin{tikzpicture}
    \begin{axis}[
       width=\columnwidth,
       height=0.5\columnwidth,
	    xlabel	= {$t \, (\mathrm{s})$},
            ylabel = {$\chi^2$},
      ticklabel style = {font=\tiny},
      label style ={font=\scriptsize},
scaled y ticks=base 10:-3,
	legend cell align = left,
        legend style = {
            font = \tiny, 
	    legend columns = 2,
	   at={(axis cs:78,6900)},
	    anchor=south
        },
    ]

    \addplot[
        mark = none,
        draw = black,
        dotted,
        line width = 1pt,
    ]
    table[
	    x = t,
	    y = lb,
    ]{../fig/COORDINATED_TURN_CHI2STATS.txt};
   \addlegendentry{lower bound}

    \addplot[
        mark = none,
        draw = black,
        dotted,
        line width = 1pt,
    ]
    table[
	    x = t,
	    y = ub,
    ]{../fig/COORDINATED_TURN_CHI2STATS.txt};
    \addlegendentry{upper bound}

    \addplot[
        mark = none,
        draw = AaltoYellow,
        solid,
        line width = 1pt,
    ]
    table[
	    x = t,
	    y = K11,
    ]{../fig/COORDINATED_TURN_CHI2STATS.txt};
    \addlegendentry{K1 (0) }

    \addplot[
        mark = none,
        draw = AaltoYellow,
        dashed,
        line width = 1pt,
    ]
    table[
	    x = t,
	    y = K21,
    ]{../fig/COORDINATED_TURN_CHI2STATS.txt};
    \addlegendentry{K2 (0) }

    \end{axis}
\end{tikzpicture}}
\caption{\scriptsize{The $\chi^2$-statistic at each time point, averaged over the Monte Carlo trials with $95\%$ confidence bands for the iterative Gaussian smoothers K1 and K2. The number of iterations after initialisation is shown in parenthesis.}}\label{fig:coordinated_turn_chi2stats}
\end{figure}

\begin{table}[t]
	\centering
	\caption{\scriptsize{The RMSE in position (POS) $(\mathrm{m})$, velocity (VEL) $(\mathrm{m/s})$, turn-rate ($\Psi$) ($10^{-3}\cdot\mathrm{rad/s}$) averaged over Monte Carlo trials, and the average $\chi^2$-statistic, for the iterative Gaussian smoothers K1 and K2. The number of iterations after initialising is shown in parenthesis.}}
	\label{tab:radar_tracking_summary}
	\begin{tabular}{|l |c c c c |}
	\hline Iterations & POS & VEL & $\Psi$ & $\chi^2$      \\
	\hline K1 (0)   & 82.88 & 13.91 & 0.658 & 299.0 \\ 
	\hline K1 (1)   & 16.52 & 1.616 & 0.444 & 6.906 \\ 
	\hline K1 (2)   & 16.44 & 1.611 & 0.445 & 6.750 \\ 
	\hline K1 (3)   & 16.44 & 1.611 & 0.445 & 6.750 \\ 
	\hline K1 (4)   & 16.44 & 1.611 & 0.445 & 6.750 \\ \hline
	\hline K2 (0)   & 86.39 & 14.31 & 0.688 & 458.6 \\ 
	\hline K2 (1)   & 16.53 & 1.618 & 0.445 & 6.943 \\ 
	\hline K2 (2)   & 16.44 & 1.611 & 0.445 & 6.750 \\ 
	\hline K2 (3)   & 16.44 & 1.611 & 0.445 & 6.750 \\ 
	\hline K2 (4)   & 16.44 & 1.611 & 0.445 & 6.750 \\ \hline
	\end{tabular} 
\end{table}

\section{Conclusion}\label{sec:conclusion}
The statistical linear regression method was generalised to obtain linear approximations to non-linear SDEs. This allowed for alternate derivation of the Type II and III smoothers \cite{SarkkaSarmavuori2013} for systems with state independent diffusion. It also lead to the derivation of the novel Type I$^*$ smoother that coincides with the Type I smoother of \cite{SarkkaSarmavuori2013} for state independent diffusions. Furthermore, this linearisation technique was used to develop a continuous-discrete analogue to the iterated Gaussian smoothers \cite{Garcia2017,Tronarp2018a,Bell1994}. The method was found to offer considerable improvements in two challenging and high-dimensional target tracking scenarios, being competitive to the variational smoother \cite{Ala-Luhtala2015}.

\section*{Acknowledgment}
Financial support by the Academy of Finland through grant \#313708 and Aalto ELEC Doctoral School is acknowledged. The authors would also like to thank Juha Ala-Luhtala for sharing his variational smoother codes.

\bibliographystyle{elsarticle-num-names}

\begin{thebibliography}{31}
\providecommand{\natexlab}[1]{#1}
\providecommand{\url}[1]{\texttt{#1}}
\providecommand{\urlprefix}{URL }
\expandafter\ifx\csname urlstyle\endcsname\relax
  \providecommand{\doi}[1]{doi:\discretionary{}{}{}#1}\else
  \providecommand{\doi}[1]{doi:\discretionary{}{}{}\begingroup
  \urlstyle{rm}\url{#1}\endgroup}\fi
\providecommand{\bibinfo}[2]{#2}

\bibitem[{Titterton and Weston(2004)}]{TittertonWeston2004}
\bibinfo{author}{D.~H. Titterton}, \bibinfo{author}{J.~L. Weston},
  \bibinfo{title}{Strapdown Inertial Navigation Technology},
  \bibinfo{publisher}{The Institute of Electrical Engineers},
  \bibinfo{year}{2004}.

\bibitem[{Crassidis and Junkins(2004)}]{CrassidisJunkins2004}
\bibinfo{author}{J.~L. Crassidis}, \bibinfo{author}{J.~L. Junkins},
  \bibinfo{title}{Optimal Estimation of Dynamic Systems},
  \bibinfo{publisher}{Chapman \& Hall/CRC}, \bibinfo{year}{2004}.

\bibitem[{Lindström et~al.(2015)Lindström, Madsen, and
  Nielsen}]{LindstromMadsenNielsen2015}
\bibinfo{author}{E.~Lindström}, \bibinfo{author}{H.~Madsen},
  \bibinfo{author}{J.~N. Nielsen}, \bibinfo{title}{Statistics for Finance},
  \bibinfo{publisher}{Chapman and Hall/{CRC}}, \bibinfo{year}{2015}.

\bibitem[{Särkkä(2013)}]{Sarkka2013}
\bibinfo{author}{S.~Särkkä}, \bibinfo{title}{{B}ayesian Filtering and
  Smoothing}, Institute of Mathematical Statistics Textbooks,
  \bibinfo{publisher}{Cambridge University Press}, \bibinfo{year}{2013}.

\bibitem[{{\O}ksendal(2003)}]{Oksendal2003}
\bibinfo{author}{B.~{\O}ksendal}, \bibinfo{title}{Stochastic Differential
  Equations - An Introduction with Applications},
  \bibinfo{publisher}{Springer}, \bibinfo{year}{2003}.

\bibitem[{Kalman and Bucy(1961)}]{KalmanBucy1961}
\bibinfo{author}{R.~Kalman}, \bibinfo{author}{R.~Bucy}, \bibinfo{title}{New
  results in linear filtering and prediction theory},
  \bibinfo{journal}{Transactions of the ASME, Journal of Basic Engineering}
  \bibinfo{volume}{83} (\bibinfo{year}{1961}) \bibinfo{pages}{95--108}.

\bibitem[{Rauch et~al.(1965)Rauch, Tung, and Striebel}]{RauchTungStriebel1965}
\bibinfo{author}{H.~Rauch}, \bibinfo{author}{F.~Tung},
  \bibinfo{author}{C.~Striebel}, \bibinfo{title}{Maximum likelihood estimates
  of linear dynamic systems}, \bibinfo{journal}{AIAA Journal}
  \bibinfo{volume}{3}~(\bibinfo{number}{8}) (\bibinfo{year}{1965})
  \bibinfo{pages}{1445--1450}.

\bibitem[{Striebel(1965)}]{Striebel1965}
\bibinfo{author}{C.~T. Striebel}, \bibinfo{title}{Partial Differential
  Equations for the Conditional Distribution of a Markov Process Given Noisy
  Observations}, \bibinfo{journal}{Journal of mathematical analysis and
  applications} \bibinfo{volume}{11} (\bibinfo{year}{1965})
  \bibinfo{pages}{151--159}.

\bibitem[{Jazwinski(1970)}]{Jazwinski:1970}
\bibinfo{author}{A.~H. Jazwinski}, \bibinfo{title}{Stochastic Processes and
  Filtering Theory}, \bibinfo{publisher}{Academic Press}, \bibinfo{year}{1970}.

\bibitem[{Särkkä(2010)}]{Sarkka2010}
\bibinfo{author}{S.~Särkkä}, \bibinfo{title}{Continuous-time and
  continuous-discrete-time unscented {R}auch-{T}ung-{S}triebel smoothers},
  \bibinfo{journal}{Signal Processing} \bibinfo{volume}{90}
  (\bibinfo{year}{2010}) \bibinfo{pages}{225--235}.

\bibitem[{Särkkä and Sarmavuori(2013)}]{SarkkaSarmavuori2013}
\bibinfo{author}{S.~Särkkä}, \bibinfo{author}{J.~Sarmavuori},
  \bibinfo{title}{{G}aussian filtering and smoothing for continuous-discrete
  dynamic systems}, \bibinfo{journal}{Signal Processing} \bibinfo{volume}{93}
  (\bibinfo{year}{2013}) \bibinfo{pages}{500--510}.

\bibitem[{Leondes et~al.(1970)Leondes, Peller, and Stear}]{Leondes1970}
\bibinfo{author}{C.~T. Leondes}, \bibinfo{author}{J.~B. Peller},
  \bibinfo{author}{E.~B. Stear}, \bibinfo{title}{Nonlinear smoothing theory},
  \bibinfo{journal}{IEEE Transactions on system science and cybernetics}
  \bibinfo{volume}{6}~(\bibinfo{number}{1}) (\bibinfo{year}{1970})
  \bibinfo{pages}{63--71}.

\bibitem[{Koyama(2018)}]{Koyama2018}
\bibinfo{author}{S.~Koyama}, \bibinfo{title}{Projection smoothing for
  continuous and continuous-discrete stochastic dynamic systems},
  \bibinfo{journal}{Signal Processing} \bibinfo{volume}{144}
  (\bibinfo{year}{2018}) \bibinfo{pages}{333--340}.

\bibitem[{Brigo et~al.(1998)Brigo, Hanzon, and LeGland}]{Brigo1998}
\bibinfo{author}{D.~Brigo}, \bibinfo{author}{B.~Hanzon},
  \bibinfo{author}{F.~LeGland}, \bibinfo{title}{A differential geometric
  approach to nonlinear filtering: the projection filter},
  \bibinfo{journal}{IEEE Transactions on Automatic Control}
  \bibinfo{volume}{43}~(\bibinfo{number}{2}) (\bibinfo{year}{1998})
  \bibinfo{pages}{247--252}.

\bibitem[{Brigo et~al.(1999)Brigo, Hanzon, Le~Gland et~al.}]{Brigo1999}
\bibinfo{author}{D.~Brigo}, \bibinfo{author}{B.~Hanzon},
  \bibinfo{author}{F.~Le~Gland}, et~al., \bibinfo{title}{Approximate nonlinear
  filtering by projection on exponential manifolds of densities},
  \bibinfo{journal}{Bernoulli} \bibinfo{volume}{5}~(\bibinfo{number}{3})
  (\bibinfo{year}{1999}) \bibinfo{pages}{495--534}.

\bibitem[{Archambeau et~al.(2008)Archambeau, Opper, Shen, Cornford, and
  Shawe-taylor}]{Archambeau2007}
\bibinfo{author}{C.~Archambeau}, \bibinfo{author}{M.~Opper},
  \bibinfo{author}{Y.~Shen}, \bibinfo{author}{D.~Cornford},
  \bibinfo{author}{J.~S. Shawe-taylor}, \bibinfo{title}{Variational Inference
  for Diffusion Processes}, in: \bibinfo{editor}{J.~C. Platt},
  \bibinfo{editor}{D.~Koller}, \bibinfo{editor}{Y.~Singer},
  \bibinfo{editor}{S.~T. Roweis} (Eds.), \bibinfo{booktitle}{Advances in Neural
  Information Processing Systems 20}, \bibinfo{publisher}{Curran Associates,
  Inc.}, \bibinfo{pages}{17--24}, \bibinfo{year}{2008}.

\bibitem[{Ala-Luhtala et~al.(2015)Ala-Luhtala, Särkkä, and
  Piché}]{Ala-Luhtala2015}
\bibinfo{author}{J.~Ala-Luhtala}, \bibinfo{author}{S.~Särkkä},
  \bibinfo{author}{R.~Piché}, \bibinfo{title}{Gaussian filtering and
  variational approximations for Bayesian smoothing in continuous-discrete
  stochastic dynamic systems}, \bibinfo{journal}{Signal Processing}
  \bibinfo{volume}{111}~(\bibinfo{number}{Supplement C}) (\bibinfo{year}{2015})
  \bibinfo{pages}{124 -- 136}.

\bibitem[{Sutter et~al.(2016)Sutter, Ganguly, and Koeppl}]{Sutter2016}
\bibinfo{author}{T.~Sutter}, \bibinfo{author}{A.~Ganguly},
  \bibinfo{author}{H.~Koeppl}, \bibinfo{title}{A variational approach to path
  estimation and parameter inference of hidden diffusion processes},
  \bibinfo{journal}{The Journal of Machine Learning Research}
  \bibinfo{volume}{17}~(\bibinfo{number}{1}) (\bibinfo{year}{2016})
  \bibinfo{pages}{6544--6580}.

\bibitem[{Cseke et~al.(2016)Cseke, Schnoerr, Opper, and
  Sanguinetti}]{Cseke2016}
\bibinfo{author}{B.~Cseke}, \bibinfo{author}{D.~Schnoerr},
  \bibinfo{author}{M.~Opper}, \bibinfo{author}{G.~Sanguinetti},
  \bibinfo{title}{Expectation propagation for continuous time stochastic
  processes}, \bibinfo{journal}{Journal of Physics A: Mathematical and
  Theoretical} \bibinfo{volume}{49}~(\bibinfo{number}{49})
  (\bibinfo{year}{2016}) \bibinfo{pages}{494002}.

\bibitem[{García{-}Fernández et~al.(2017)García{-}Fernández, Svensson, and
  S\"arkk\"a}]{Garcia2017}
\bibinfo{author}{{\'A}.~F. García{-}Fernández},
  \bibinfo{author}{L.~Svensson}, \bibinfo{author}{S.~S\"arkk\"a},
  \bibinfo{title}{Iterated posterior linearisation smoother},
  \bibinfo{journal}{IEEE Transactions on Automatic Control}
  \bibinfo{volume}{62}~(\bibinfo{number}{4}) (\bibinfo{year}{2017})
  \bibinfo{pages}{2056--2063}.

\bibitem[{Tronarp et~al.(2018)Tronarp, Garcia-Fernandez, and
  S{\"a}rkk{\"a}}]{Tronarp2018a}
\bibinfo{author}{F.~Tronarp}, \bibinfo{author}{A.~F. Garcia-Fernandez},
  \bibinfo{author}{S.~S{\"a}rkk{\"a}}, \bibinfo{title}{Iterative Filtering and
  Smoothing In Non-Linear and Non-{G}aussian Systems Using Conditional
  Moments}, \bibinfo{journal}{IEEE Signal Processing Letters}
  \bibinfo{volume}{25}~(\bibinfo{number}{3}) (\bibinfo{year}{2018})
  \bibinfo{pages}{408--412}.

\bibitem[{Bell(1994)}]{Bell1994}
\bibinfo{author}{B.~M. Bell}, \bibinfo{title}{The iterated {K}alman smoother as
  a {G}auss--{N}ewton method}, \bibinfo{journal}{SIAM Journal on Optimization}
  \bibinfo{volume}{4}~(\bibinfo{number}{3}) (\bibinfo{year}{1994})
  \bibinfo{pages}{626--636}.

\bibitem[{Lefebvre et~al.(2002)Lefebvre, Bruyninckx, and
  De~Schuller}]{Lefebvre:2002}
\bibinfo{author}{T.~Lefebvre}, \bibinfo{author}{H.~Bruyninckx},
  \bibinfo{author}{J.~De~Schuller}, \bibinfo{title}{Comment on ``{A} new method
  for the nonlinear transformation of means and covariances in filters and
  estimators'' [with authors' reply]}, \bibinfo{journal}{IEEE Transactions on
  Automatic Control} \bibinfo{volume}{47}~(\bibinfo{number}{8})
  (\bibinfo{year}{2002}) \bibinfo{pages}{1406--1409}.

\bibitem[{Kallianpur and Striebel(1968)}]{Kallianpur1968}
\bibinfo{author}{G.~Kallianpur}, \bibinfo{author}{C.~Striebel},
  \bibinfo{title}{Estimation of stochatic systems: arbitrary system process
  with additive white noise observation errors}, \bibinfo{journal}{Annals of
  Mathematical Statistics} \bibinfo{volume}{39}~(\bibinfo{number}{3})
  (\bibinfo{year}{1968}) \bibinfo{pages}{785--801}.

\bibitem[{García{-}Fernández et~al.(2015)García{-}Fernández, Svensson,
  Morelande, and S\"arkk\"a}]{Garcia2015}
\bibinfo{author}{{\'A}.~F. García{-}Fernández},
  \bibinfo{author}{L.~Svensson}, \bibinfo{author}{M.~R. Morelande},
  \bibinfo{author}{S.~S\"arkk\"a}, \bibinfo{title}{Posterior Linearization
  Filter: Principles and Implementation Using Sigma Points},
  \bibinfo{journal}{IEEE Transactions on Signal Processing}
  \bibinfo{volume}{63}~(\bibinfo{number}{20}) (\bibinfo{year}{2015})
  \bibinfo{pages}{5561--5573}.

\bibitem[{Weinstock(1974)}]{Weinstock1974}
\bibinfo{author}{R.~Weinstock}, \bibinfo{title}{Calculus of Variations},
  \bibinfo{publisher}{Dover Publications}, \bibinfo{year}{1974}.

\bibitem[{Kloeden and Platen(1999)}]{KloedenPlaten1999}
\bibinfo{author}{P.~E. Kloeden}, \bibinfo{author}{E.~Platen},
  \bibinfo{title}{Numerical Solutions of Stochastic Differential Equations},
  \bibinfo{publisher}{Springer-Verlag Berlin Heidelberg}, \bibinfo{year}{1999}.

\bibitem[{Bell and Cathey(1993)}]{Bell1993}
\bibinfo{author}{B.~M. Bell}, \bibinfo{author}{F.~W. Cathey},
  \bibinfo{title}{The iterated {K}alman filter update as a {G}auss--{N}ewton
  method}, \bibinfo{journal}{IEEE Transaction on Automatic Control}
  \bibinfo{volume}{38}~(\bibinfo{number}{2}) (\bibinfo{year}{1993})
  \bibinfo{pages}{294--297}.

\bibitem[{Axelsson and Gustafsson(2015)}]{Axelsson2015}
\bibinfo{author}{P.~Axelsson}, \bibinfo{author}{F.~Gustafsson},
  \bibinfo{title}{Discrete-time solutions to the continuous-time differential
  {L}yapunov equation with applications to Kalman filtering},
  \bibinfo{journal}{IEEE Transactions on Automatic Control}
  \bibinfo{volume}{60}~(\bibinfo{number}{3}) (\bibinfo{year}{2015})
  \bibinfo{pages}{632--643}.

\bibitem[{Arasaratnam and Haykin(2009)}]{Arasaratnam2009}
\bibinfo{author}{I.~Arasaratnam}, \bibinfo{author}{S.~Haykin},
  \bibinfo{title}{Cubature {K}alman filters}, \bibinfo{journal}{IEEE
  Transactions on Automatic Control} \bibinfo{volume}{54}~(\bibinfo{number}{6})
  (\bibinfo{year}{2009}) \bibinfo{pages}{1254--1269}.

\bibitem[{Bar-Shalom et~al.(2001)Bar-Shalom, Li, and
  Kirubarajan}]{BarShalom2004}
\bibinfo{author}{Y.~Bar-Shalom}, \bibinfo{author}{X.~R. Li},
  \bibinfo{author}{T.~Kirubarajan}, \bibinfo{title}{Estimation with
  applications to tracking and navigation: theory algorithms and software},
  \bibinfo{publisher}{John Wiley \& Sons}, \bibinfo{year}{2001}.

\end{thebibliography}

\end{document}